\documentclass{IEEEtran}
\usepackage[T1]{fontenc}
\usepackage[latin9]{inputenc}
\usepackage{color}
\usepackage{amsthm}
\usepackage{amsmath}
\usepackage{amssymb}
\usepackage{graphicx}
\usepackage{esint}

\makeatletter

\providecommand{\tabularnewline}{\\}

  \theoremstyle{plain}
  \newtheorem{lem}{\protect\lemmaname}
\theoremstyle{plain}
\newtheorem{thm}{\protect\theoremname}
  \theoremstyle{plain}
  \newtheorem{cor}{\protect\corollaryname}
  \theoremstyle{plain}
  \newtheorem{prop}{\protect\propositionname}

\usepackage{cite}
\usepackage{wrapfig}
\IEEEoverridecommandlockouts

\makeatother

  \providecommand{\lemmaname}{Lemma}
  \providecommand{\propositionname}{Proposition}
\providecommand{\corollaryname}{Corollary}
\providecommand{\theoremname}{Theorem}

\begin{document}

\title{Downlink Analysis for a Heterogeneous Cellular Network}

\author{{\normalsize Prasanna Madhusudhanan$^{*}$, Juan G. Restrepo$^{\dagger}$,
Youjian (Eugene) Liu$^{*}$, Timothy X Brown}$^{*+}${\normalsize \\$^{*}$
Department of Electrical, Computer and Energy Engineering, $^{\dagger}$
Department of Applied Mathematics, $^{+}$ Interdisciplinary Telecommunications
Program\\University of Colorado, Boulder, CO 80309-0425 USA\\\{mprasanna,
juanga, eugeneliu, timxb\}@colorado.edu}%
\thanks{A special cases of the results in this paper were presented in \cite{Madhusudhanan2011,Madhusudhanan2012,Madhusudhanan2012c}%
}}
\maketitle
\begin{abstract}
In this paper, a comprehensive study of the the downlink performance
in a heterogeneous cellular network (or hetnet) is conducted. A general
hetnet model is considered consisting of an arbitrary number of open-access
and closed-access tier of base stations (BSs) arranged according to
independent homogeneous Poisson point processes. The BSs of each tier
have a constant transmission power, random fading coefficient with
an arbitrary distribution and arbitrary path-loss exponent of the
power-law path-loss model. For such a system, analytical characterizations
for the coverage probability and average rate at an arbitrary mobile-station
(MS), and average per-tier load are derived for both the max-SINR
connectivity and nearest-BS connectivity models. Using stochastic
ordering, interesting properties and simplifications for the hetnet
downlink performance are derived by relating these two connectivity
models to the maximum instantaneous received power (MIRP) connectivity
model and the maximum biased received power (MBRP) connectivity models,
respectively, providing good insights about the hetnets and the downlink
performance in these complex networks. Furthermore, the results also
demonstrate the effectiveness and analytical tractability of the stochastic
geometric approach to study the hetnet performance.\end{abstract}

\section{Introduction\label{sec:Introduction}}

\IEEEPARstart{T}{he} modern cellular communication network is an
overlay of multiple contributing subnetworks such as the macrocell,
microcell, picocell and femtocell networks, collectively called the
heterogeneous network (or, in short, \textit{hetnets}). The hetnets
have been shown to sustain greater end-user data-rates and throughput
as well as provide indoor and cell-edge coverage, further leading
to their inclusion as an important feature to be implemented under
the fourth-generation (4G) cellular standards \cite{Yeh2011,Damnjanovic2011,Damnjanovic2012,Report2011,Qualcomm2010,Chandrasekhar2008,Lagrange1997}.

Until recently, the analysis of such networks has been done solely
through system simulations. Since the hetnets consist of a combination
of regularly spaced macrocell base-stations (BSs) along with irregularly
spaced microcell and picocell BSs and often randomly placed end-user
deployed femtocell BSs, it is difficult to study the entire network
at once using simulations. Further, the BSs in each of these networks
have different transmission powers, traffic-load carrying capabilities
and different radio environment that is based on the locations in
which they are deployed. The many parameters involved in the design
and modeling of the individual networks makes it difficult to narrow
all the possibilities down to a limited set of simulation scenarios
based on which one can make design decisions for the entire network.
Under these circumstances, the development of an analytical model
that captures all the design scenarios of interest is of great importance.

Towards this goal, a stochastic geometric model has been identified
as a plausible analytical model as well as the most widely used one
in academia. There is a rich set of results for studying the behavior
of large systems with nodes deployed randomly (especially according
to a homogeneous Poisson point process on the plane) and can be found
in \cite{Haenggi2008,Baccelli2009,Baccelli2009a,Nguyen2010}. For
the cellular network, a strong motivation for viewing the BS arrangement
as a homogeneous Poisson point process can be drawn from the study
of the cellular systems in \cite{Brown1998,Brown2000,Andrews2011}
which suggests that significant insights can be gained by bounding
the downlink cellular performance between the ideal hexagonal grid
model and the homogeneous Poisson point process based model. More
interestingly, in \cite[Fig.2.]{Brown2000}, it is claimed with the
help of Monte-Carlo simulations that in the limit of strong log-normal
shadow fading (standard deviation of the fading coefficient $\sigma\rightarrow\infty$),
the downlink performance of an ideal hexagonal cellular system approaches
the performance in a cellular system with randomly deployed base-stations
according to a homogeneous Poisson point process. Recently, the above
convergence has been analytically proved in \cite[Theorem 3]{Blaszczyszyn}.
It is shown that the downlink performance of a cellular network with
any deterministic arrangement of BSs (not just the ideal hexagonal
grid model) converges to that of a Poisson point process based model
as $\sigma\rightarrow\infty$, and moreover even for realistic values
of $\sigma$ that are observed in the indoor environments, the latter
model is a good approximation for the deterministic model. Results
in \cite{Madh0000:Carrier,Madhusudhanan2010a,Madhusudhanan2012a}
demonstrate that, with the Poisson point process based BS arrangement,
the study of the cellular system has the distinctive advantage of
being analytically tractable, unlike the studies based on the hexagonal
grid model that are purely simulation-based. 

In light of the above motivations, it is well-justified to study the
hetnet performance by viewing the hetnet as composed of multiple tiers
of networks (e.g. macrocell, microcell, picocell and femtocell networks),
each modeled as an independent homogeneous Poisson point process,
and such studies have been done in \cite{Mukherjee2012a,Mukherjee2012,Mukherjee2011,Dhillon2012a,Dhillon2011b,Dhillon2011a,Dhillon2012,Jo2011}
and by us in \cite{Madhusudhanan2011,Madhusudhanan2012,Madhusudhanan2012c}.
These studies mathematically characterize important performance metrics
such as coverage probability (1 - outage probability), average ergodic
rate, average load carried by BSs of each tier and load-awareness.
Furthermore, such studies have facilitated the characterization of
the improvements that techniques such as fractional frequency reuse
and carrier aggregation bring to cellular performance as well as hetnet
performance. In the following subsection, we differentiate our work
from the other prior work on hetnets and list the contributions of
this paper.

\subsection*{Contributions of the paper}

Here, the hetnet is modeled to consist of open and closed access networks
formed by the arrangement of BSs according to homogeneous Poisson
point process with a certain density for each tier, and independent
of the other tiers. The focus is on the downlink performance analysis;
the MS has access to only the open-access tiers and connects to one
of the BSs in these tiers. The closed access tiers only cause interference
at the MS. Hence, we study the downlink performance where the hetnet
consists of an arbitrary number of open and closed access tiers. Signals
from BSs of a given tier have a constant transmit power, random fading
coefficient that is i.i.d. across all the BSs of the same tier and
independent of those of the other tiers with any arbitrary distribution,
arbitrary path-loss exponent that is constant for all BSs of the same
tier and different across different tiers, and the signal-to-interference-plus-noise-ratio
(SINR) threshold for connectivity to a given $k^{\mathrm{th}}$ open-access
tier's BS is $\beta_{k},\ k=1,\cdots,K$. For such a general setting,
expressions for the coverage probability at the MS are derived for
both the max-SINR connectivity model and the nearest-BS connectivity
model. In the former connectivity model, the MS is said to be in coverage
if there exists at least one open-access BS with an SINR above the
corresponding threshold, and under the latter connectivity model,
the MS is said to be in coverage if at least one among the nearest
BSs of each open-access tier has an SINR above the corresponding threshold. 

The results shown here are generalizations of the existing results
in \cite{Dhillon2011a,Jo2011,Mukherjee2012,Madhusudhanan2012c,Madhusudhanan2012}.
In \cite{Dhillon2011b,Dhillon2012a}, the coverage probability results
are obtained for the hetnets under the max-SINR connectivity, but
for the case where the fading coefficients for the BS transmissions
are independent and identically distributed (i.i.d.) exponential random
variables, and the path-loss exponents are the same for all the tiers.
Using an entirely different approach, \cite{Mukherjee2011,Mukherjee2012,Mukherjee2012a}
derives the coverage probability for the hetnet with max-SINR and
nearest BS connectivity models, but were again restricted to the i.i.d.
exponential distribution case for fading. In \cite{Jo2011}, the authors
study the hetnet coverage probability for the maximum biased received
power (MBRP) connectivity model (which is a special case of the nearest-BS
connectivity model, as will be seen later), and again for the exponential
fading assumption for all BS transmissions. In \cite{Madhusudhanan2011,Madhusudhanan2012,Madhusudhanan2012c},
we derived the hetnet coverage probability for the case when the i.i.d.
fading coefficients have an arbitrary distribution and the path-loss
exponents are different for different tiers, for the maximum instantaneous
received power (MIRP) connectivity model, which is a special case
for the max-SINR connectivity model, as will be discussed later. Here,
we derive the coverage probabilities for the general connectivity
models (max-SINR and nearest-BS) for the general system settings mentioned
above. 

When the SINR thresholds of all the tiers are above 1, the hetnet
coverage probability under max-SINR connectivity and MIRP connectivity
are identical, and nearest-BS connectivity and the MBRP connectivity
are identical. Further, in these special cases, simple analytical
expression are derived for the coverage probability, average rate
and the load carried by the BSs of each tier. The following section
describes the system model in detail.

\section{System Model\label{sec:System-Model}}

This section describes the various elements used to model the wireless
network, namely, the BS layout, the radio environment, and the performance
metrics of interest.

\subsubsection{BS Layout\label{sub:BS-Layout}}

The hetnet is composed of $K$ open-access tiers and $L$ closed-access
tier, and the BS layout in each tier is according to an independent
homogeneous Poisson point process in $\mathbb{R}^{2}$ with density
$\lambda_{ok}$, $\lambda_{cl}$ for the $k^{\mathrm{th}}$ open-access
tier and $l^{\mathrm{th}}$ closed-access tier, respectively, where
$k=1,\dots,\ K$ and $l=1,\dots,\ L$. The MS is allowed to communicate
with any BS of the open-access tiers, but cannot communicate with
any of the closed-access BSs.

\subsubsection{Radio Environment and downlink SINR\label{sub:Radio-Environment}}

The signal transmitted from each BS undergoes shadow fading and path-loss.
The SINR at an arbitrary MS in the system from the $i^{\mathrm{th}}$
BS of the $k^{\mathrm{th}}$ open-access tier is the ratio of the
received power from this BS to the sum of the interferences from all
the other BSs in the system and the constant background noise $\eta$,
and is expressed as 
\begin{eqnarray}
\mathrm{SINR}_{ki} & = & \frac{P_{ok}\Psi_{ki}R_{ki}^{-\varepsilon_{k}}}{I_{o}-P_{ok}\Psi_{ki}R_{ki}^{-\varepsilon_{k}}+I_{c}+\eta},\label{eq:SINRDefinition}
\end{eqnarray}
where $I_{o}=\sum_{m=1}^{K}\sum_{n=1}^{\infty}P_{om}\Psi_{mn}R_{mn}^{-\varepsilon_{m}}$
is the sum of the received powers from all the open-access tier BSs,
$\left\{ P_{om},\Psi_{mn},\varepsilon_{m},R_{mn}\right\} _{m=1,\ n=1}^{m=K,\ n=\infty}$
are the constant transmit power, random shadow fading factor, constant
path-loss exponent and the distance from the MS of the $n^{\mathrm{th}}$
BS of the $m^{\mathrm{th}}$ open-access tier. Similarly, $I_{c}=\sum_{l=1}^{L}\sum_{n=1}^{\infty}P_{cl}\Psi_{cln}R_{cln}^{-\varepsilon_{cl}}$
is the sum of the received powers from all the closed-access tier
BSs, $\left\{ P_{cl},\Psi_{cln},\varepsilon_{cl},R_{cln}\right\} _{l=1,\ n=1}^{l=L,\ n=\infty}$
lists the constant transmit power, random shadow fading factor, and
the constant path-loss exponent of the $n^{\mathrm{th}}$ BS of the
$l^{\mathrm{th}}$ closed-access tier. The fading coefficients $\left\{ \Psi_{mn}\right\} _{n=1}^{\infty}$
$\left(\left\{ \Psi_{cln}\right\} _{n=1}^{\infty}\right)$ are i.i.d.
random variables with the same distribution as $\Psi_{m}$ $\left(\Psi_{cl}\right)$,
$m=1,\dots,\ K$ $\left(l=1,\dots,\ L\right)$. Further, following
\cite{Madhusudhanan2010a}, it is assumed that $\left\{ \mathbb{E}\left[\Psi_{m}^{\frac{2}{\varepsilon_{m}}}\right]\right\} _{m=1}^{K},\ \left\{ \mathbb{E}\left[\Psi_{cl}^{\frac{2}{\varepsilon_{cl}}}\right]\right\} _{l=1}^{L}<\infty$.
Finally, $R_{mn}$ $\left(R_{cln}\right)$ is the distance of the
$n^{\mathrm{th}}$ nearest BS belonging to the $m^{\mathrm{th}}$
open-access ( $l^{\mathrm{th}}$ closed-access) tier, and $\left\{ R_{mn}\right\} _{n=1}^{\infty},\ \left\{ R_{cln}\right\} _{n=1}^{\infty}$
represents the distance from origin of the sets of points distributed
according to the homogeneous Poisson point processes described in
Section \ref{sub:BS-Layout}. The various symbols introduced in this
section are listed in Table \ref{tab:symbolList} for quick reference.
\begin{table*}
\begin{centering}
\begin{tabular}{|c|c|}
\hline 
Symbol & Description\tabularnewline
\hline 
$K,\ L$ & Number of open-access and closed-access tiers, respectively.\tabularnewline
\hline 
$\left\{ \lambda_{ok}\right\} _{k=1}^{K},\ \left\{ \lambda_{cl}\right\} _{l=1}^{L}$ & BS densities of open-access and closed-access tiers, respectively.\tabularnewline
\hline 
$\left\{ P_{ok}\right\} _{k=1}^{K},\ \left\{ P_{cl}\right\} _{l=1}^{L}$ & Constant transmission powers of the BSs \tabularnewline
 & of the K open-access tiers and closed access tier, respectively\tabularnewline
\hline 
$\left\{ \varepsilon_{k}\right\} _{k=1}^{K},\ \left\{ \varepsilon_{cl}\right\} _{l=1}^{L}$ & Path-loss exponents of the open and closed - access tiers ( > 2). \tabularnewline
\hline 
$\left\{ \Psi_{k}\right\} _{k=1}^{K},\left\{ \Psi_{cl}\right\} _{l=1}^{L}$ & i.i.d. fading gains of the open and closed-access tiers $\left(\mathbb{E}\Psi_{k}^{\frac{2}{\varepsilon_{l}}},\ \mathbb{E}\Psi_{cl}^{\frac{2}{\varepsilon_{cl}}}<\infty\right)$\tabularnewline
\hline 
$\left\{ \beta_{k}\right\} _{k=1}^{K}$ & SINR thresholds for connectivity to a BS in the $k^{\mathrm{th}}$
open-access tier\tabularnewline
\hline 
$\eta$ & Background noise power\tabularnewline
\hline 
$\left\{ \gamma_{k}\right\} _{k=1}^{K}$ & =$\left\{ 1+\frac{1}{\beta_{k}}\right\} _{k=1}^{K}$\tabularnewline
\hline 
\end{tabular}
\par\end{centering}

\caption{\label{tab:symbolList}List of symbols used in the paper}
\vspace{-0.2in}
\end{table*}

\subsubsection{BS connectivity models\label{sub:Cell-Association-policy}}

A MS is able to communicate with a BS of the $k^{\mathrm{th}}$ open-access
tier if the corresponding SINR is above a certain threshold $\beta_{k},\ k=1,\cdots,\ K$.
In this case, the MS is said to be in coverage. The BS connectivity
models provide a rule to determine which BS to connect to, and in
this paper, we focus on the max-SINR connectivity model and the nearest-BS
connectivity model. The MIRP connectivity model and the MBRP connectivity
model are special cases of the max-SINR and nearest-BS connectivity
models, respectively, and will be discussed in detail in the later
sections.

Under the max-SINR connectivity model, the MS is said to be in coverage
if there exists at-least one BS among all the open-access tiers with
an SINR at the MS above the corresponding threshold, and is mathematically
expressed as follows. 
\begin{eqnarray}
 &  & \mathbb{P}_{\mathrm{coverage}}^{\mathrm{max-SINR}}=\mathbb{P}\left(\bigcup_{k=1}^{K}\bigcup_{i=1}^{\infty}\left\{ \mathrm{SINR}_{ki}>\beta_{k}\right\} \right)\nonumber \\
 &  & =\mathbb{P}\left(\bigcup_{k=1}^{K}\left\{ \mathrm{SINR}_{k}\left(\max\right)>\beta_{k}\right\} \right),\label{eq:maxSINRcoverageDef}
\end{eqnarray}
where $\mathrm{SINR}_{ki}$ corresponds to the $i^{\mathrm{th}}$
BS of the $k^{\mathrm{th}}$ tier as defined in (\ref{eq:SINRDefinition})
and $\mathrm{SINR}_{k}\left(\max\right)$ is the maximum SINR at the
MS among all the $k^{\mathrm{th}}$ open-access tier BSs.

The MS is said to be in coverage under the nearest-BS connectivity
model if there exists at least one of the nearest BSs of the $K$
open-access tiers with SINR at the MS above the corresponding threshold.
This is mathematically expressed as 
\begin{eqnarray}
\mathbb{P}_{\mathrm{coverage}}^{\mathrm{nearest}} & = & \mathbb{P}\left(\bigcup_{k=1}^{K}\left\{ \mathrm{SINR}_{k1}>\beta_{k}\right\} \right),\label{eq:nearestBScoverageDef}
\end{eqnarray}
where $\mathrm{SINR}_{k1}$ (see (\ref{eq:SINRDefinition})) is the
SINR at the MS from the nearest BS among the $k^{\mathrm{th}}$ tier
BSs. In the following section, we derive expressions for the hetnet
coverage probability for the above mentioned connectivity models.

\section{\label{sec:hetnetCoverage}Hetnet Coverage Probability}

In \cite{Nguyen2010}, a technique to compute the downlink coverage
probability under max-SINR connectivity for a single-tier network
was shown. In \cite{Dhillon2011b}, this technique is used to compute
the hetnet coverage probability for an open-access case where the
fading coefficients for all the BSs in the system are i.i.d. unit
mean exponential random variables and the path-loss exponents are
the same for all tiers. Here, we generalize the technique developed
in \cite{Nguyen2010} to compute the hetnet coverage probability for
both the max-SINR and nearest-BS connectivity models for a general
system model explained in Section \ref{sec:System-Model}. 

The coverage probability expressions in (\ref{eq:maxSINRcoverageDef})
and (\ref{eq:nearestBScoverageDef}) can be equivalently expressed
as follows: 
\begin{eqnarray}
 &  & \mathbb{P}_{\mathrm{coverage}}^{\mathrm{max-SINR}}=\mathbb{P}\left(\bigcup_{k=1}^{K}\left\{ \frac{M_{k}}{I_{o}+I_{c}+\eta-M_{k}}>\beta_{k}\right\} \right)\nonumber \\
 &  & =\mathbb{P}\left(\left\{ \underset{k=1,\cdots,\ K}{\max}\gamma_{k}M_{k}>I_{o}+I_{c}+\eta\right\} \right),\label{eq:maxSINRcoverageEqexpr}\\
 &  & \mathbb{P}_{\mathrm{coverage}}^{\mathrm{nearest}}=\mathbb{P}\left(\bigcup_{k=1}^{K}\left\{ \frac{N_{k}}{I_{o}+I_{c}+\eta-N_{k}}>\beta_{k}\right\} \right)\nonumber \\
 &  & =\mathbb{P}\left(\left\{ \underset{k=1,\cdots,\ K}{\max}\gamma_{k}N_{k}>I_{o}+I_{c}+\eta\right\} \right),\label{eq:nearestBScoverageEqexpr}
\end{eqnarray}
where $M_{k}=\underset{n=1,\cdots,\ \infty}{\max}P_{ok}\Psi_{okl}R_{kl}^{-\varepsilon_{k}}$
is the maximum of the received powers from all the $k^{\mathrm{th}}$
tier BSs, $N_{k}=P_{k}\Psi_{k1}R_{k1}^{-\varepsilon_{k}}$ is the
received power from the nearest BS among all the $k^{\mathrm{th}}$
tier BSs, $I_{o}$ $\left(I_{c}\right)$ is the sum of the received
powers from all the open-access BSs (closed-access BSs) in the system,
and are defined in (\ref{eq:SINRDefinition}). We begin with computing
the Laplace transform of the interference from the closed-access tiers,
$I_{c}$, $\mathcal{L}_{I_{c}}\left(s\right)=\mathbb{E}\left[\mathrm{e}^{-sI_{c}}\right]$. 
\begin{lem}
\label{lem:LTIcExpression}The Laplace transform of the interference
from the closed-access tiers is 
\begin{eqnarray}
\mathcal{L}_{I_{c}}\left(s\right) & = & \mathrm{e}^{-\sum_{l=1}^{L}\lambda_{cl}\pi\left(sP_{cl}\right)^{\frac{2}{\varepsilon_{cl}}}\mathbb{E}\left[\Psi_{cl}^{\frac{2}{\varepsilon_{cl}}}\right]\Gamma\left(1-\frac{2}{\varepsilon_{cl}}\right)}.\label{eq:LTIcExpression}
\end{eqnarray}
\end{lem}
\begin{proof}
The proof for (\ref{eq:LTIcExpression}) is as follows. $\mathcal{L}_{I_{c}}\left(s\right)=\mathbb{E}\left[\exp\left(-s\sum_{l=1}^{L}\sum_{n=1}^{\infty}P_{cl}\Psi_{cln}R_{cln}^{-\varepsilon_{cl}}\right)\right]\overset{\left(a\right)}{=}\prod_{l=1}^{L}\mathbb{E}\left[\exp\left(-s\sum_{n=1}^{\infty}P_{cl}\Psi_{cln}R_{cln}^{-\varepsilon_{cl}}\right)\right]\overset{\left(b\right)}{=}\prod_{l=1}^{L}\exp\left(-\lambda_{cl}\mathbb{E}_{\Psi_{cl}}\left[\int_{r=0}^{\infty}\left(1-\mathrm{e}^{-sP_{cl}\Psi_{cl}r^{-\varepsilon_{cl}}}\right)2\pi rdr\right]\right),$
where $\left(a\right)$ is obtained because the BS arrangement for
the $L$ closed-access tiers and the corresponding transmission and
fading parameters are independent of each other, and $\left(b\right)$
evaluates the expectation in $\left(a\right)$ using the Campbell's
theorem of Poisson point process \cite[Page 28]{Kingman1993}, and
(\ref{eq:LTIcExpression}) is obtained by evaluating the integral
in $\left(b\right)$.
\end{proof}
Next, we derive expressions for two Laplace transforms that are useful
to obtain semi-analytical expressions for $\mathbb{P}_{\mathrm{coverage}}^{\mathrm{max-SINR}}$
and $\mathbb{P}_{\mathrm{coverage}}^{\mathrm{nearest}}$, respectively.
\begin{lem}
\label{lem:LTmaxSINRnearestBS}
\begin{eqnarray}
 &  & \mathcal{L}_{I_{o}+I_{c}+\eta,\ \underset{k=1,\cdots,K}{\max}\gamma_{k}M_{k}\le u}\left(s\right)\nonumber \\
 &  & \triangleq\mathbb{E}\left[\mathrm{e}^{-s\left(I_{o}+I_{c}+\eta\right)}\mathcal{I}\left(\underset{k=1,\cdots,K}{\max}\gamma_{k}M_{k}\le u\right)\right]\nonumber \\
 &  & =\mathcal{L}_{I_{c}}\left(s\right)\exp\left(-s\eta-\sum_{k=1}^{K}\lambda_{ok}\pi\left(sP_{ok}\right)^{\frac{2}{\varepsilon_{k}}}\mathbb{E}\left[\Psi_{k}^{\frac{2}{\varepsilon_{k}}}\right]\times\right.\nonumber \\
 &  & \left.\left[\Gamma\left(1-\frac{2}{\varepsilon_{k}}\right)+\frac{2}{\varepsilon_{k}}\Gamma\left(-\frac{2}{\varepsilon_{k}},\frac{su}{\gamma_{k}}\right)\right]\right),\label{eq:maxSINRLTexp}\\
 &  & \mathcal{L}_{I_{o}+I_{c}+\eta,\ \underset{k=1,\cdots,K}{\max}\gamma_{k}N_{k}\le u}\left(s\right)\nonumber \\
 &  & \triangleq\mathbb{E}\left[\exp\left(-s\left(I_{o}+I_{c}+\eta\right)\right)\times\mathcal{I}\left(\underset{k=1,\cdots,K}{\max}\gamma_{k}N_{k}\le u\right)\right]\nonumber \\
 &  & =\mathcal{L}_{I_{c}}\left(s\right)\mathrm{e}^{-s\eta-\sum_{k=1}^{K}\lambda_{ok}\pi\left(sP_{ok}\right)^{\frac{2}{\varepsilon_{k}}}\mathbb{E}\left[\Psi_{k}^{\frac{2}{\varepsilon_{k}}}\right]\Gamma\left(1-\frac{2}{\varepsilon_{k}}\right)}\times\nonumber \\
 &  & \prod_{k=1}^{K}\mathbb{E}_{\Psi_{k1}}\left[\int_{x=0}^{\frac{su}{\gamma_{k}\Psi_{k1}}}\lambda_{ok}\frac{2\pi}{\varepsilon_{k}}\left(sP_{ok}\right)^{\frac{2}{\varepsilon_{k}}}x^{-\frac{2}{\varepsilon_{k}}-1}\times\right.\nonumber \\
 &  & \left.\mathrm{e}^{-\Psi_{k1}x-\lambda_{ok}\frac{2\pi}{\varepsilon_{k}}\left(sP_{ok}\right)^{\frac{2}{\varepsilon_{k}}}\mathbb{E}_{\Psi_{k}}\left[\Psi_{k}^{\frac{2}{\varepsilon_{k}}}\Gamma\left(-\frac{2}{\varepsilon_{k}},x\Psi_{k}\right)\right]}dx\right],\label{eq:nearestBSLTexp}
\end{eqnarray}
where $\mathcal{L}_{I_{c}}\left(s\right)$ is from Lemma \ref{lem:LTIcExpression}
and the random variables $\Psi_{k1}$ and $\Psi_{k}$ are i.i.d. for
all $k=1,\cdots,K$.\end{lem}
\begin{proof}
See Appendix \ref{sub:proofLTmaxSINRnearestBS}.
\end{proof}

The significance of Lemmas \ref{lem:LTIcExpression} and \ref{lem:LTmaxSINRnearestBS}
are as follows. Notice from (\ref{eq:maxSINRcoverageEqexpr}) and
(\ref{eq:nearestBScoverageEqexpr}) that the hetnet coverage probability
can be obtained if the joint probability density function (p.d.f.)
of $\left(I_{o}+I_{c}+\eta,\ \underset{i=1,\cdots,K}{\max}\gamma_{i}M_{i}\right)$
and $\left(I_{o}+I_{c}+\eta,\ \underset{i=1,\cdots,K}{\max}\gamma_{i}N_{i}\right)$
is known. The joint p.d.f.s can be derived from the Laplace transform
expressions in Lemma \ref{lem:LTmaxSINRnearestBS} using the following
simple operations. 
\begin{eqnarray}
 &  & \hspace{-0.75cm}f_{I_{o}+I_{c}+\eta,\underset{i=1,\cdots,K}{\max}\gamma_{i}M_{i}}\left(x,y\right)=\nonumber \\
 &  & \hspace{-0.75cm}\int_{\omega=-\infty}^{\infty}\left.\frac{\partial}{\partial u}\mathcal{L}_{I_{o}+I_{c}+\eta,\ \underset{i=1,\cdots,K}{\max}\gamma_{i}M_{i}\le u}\left(j\omega\right)\right|_{u=y}\frac{\mathrm{e}^{j\omega x}}{2\pi}d\omega,\label{eq:jointpdfmaxSINR}\\
 &  & \hspace{-0.75cm}f_{I_{o}+I_{c}+\eta,\underset{i=1,\cdots,K}{\max}\gamma_{i}N_{i}}\left(x,y\right)=\nonumber \\
 &  & \hspace{-0.75cm}\int_{\omega=-\infty}^{\infty}\left.\frac{\partial}{\partial u}\mathcal{L}_{I_{o}+I_{c}+\eta,\ \underset{i=1,\cdots,K}{\max}\gamma_{i}N_{i}\le u}\left(j\omega\right)\right|_{u=y}\frac{\mathrm{e}^{j\omega x}}{2\pi}d\omega,\label{eq:jointpdfnearestBS}
\end{eqnarray}

where $f_{\cdot,\cdot}\left(\cdot,\cdot\right)$ denotes the joint
p.d.f. of the involved random variables. This is shown for the max-SINR
connectivity case in \cite[Corollary 4]{Nguyen2010}, and exactly
the same steps can be used to derive (\ref{eq:jointpdfmaxSINR}) and
(\ref{eq:jointpdfnearestBS}). It is not shown here to avoid repetition.
Further, the partial derivative terms in the above equations can be
easily computed and are given below. 
\begin{eqnarray}
 &  & \hspace{-0.75cm}\frac{\frac{\partial}{\partial u}\mathcal{L}_{I_{o}+I_{c}+\eta,\ \underset{i=1,\cdots,K}{\max}\gamma_{i}M_{i}\le u}\left(s\right)}{\mathcal{L}_{I_{o}+I_{c}+\eta,\ \underset{i=1,\cdots,K}{\max}\gamma_{i}M_{i}\le u}\left(s\right)}=\nonumber \\
 &  & \hspace{-0.75cm}\sum_{k=1}^{K}\lambda_{k}\frac{2\pi}{\varepsilon_{k}}\left(\gamma_{k}P_{k}\right)^{\frac{2}{\varepsilon_{k}}}\mathbb{E}\left[\Psi_{k}^{\frac{2}{\varepsilon_{k}}}\right]u^{-1-\frac{2}{\varepsilon_{k}}}\mathrm{e}^{-\frac{su}{\gamma_{k}}},\label{eq:derivativeLTmaxSINR}\\
 &  & \hspace{-0.75cm}\frac{\frac{\partial}{\partial u}\mathcal{L}_{I_{o}+I_{c}+\eta,\ \underset{i=1,\cdots,K}{\max}\gamma_{i}N_{i}\le u}\left(s\right)}{\mathcal{L}_{I_{o}+I_{c}+\eta,\ \underset{i=1,\cdots,K}{\max}\gamma_{i}N_{i}\le u}\left(s\right)}=\nonumber \\
 &  & \hspace{-0.75cm}\sum_{k=1}^{K}\frac{\frac{\mathbb{E}_{\Psi_{k1}}\left[\Psi_{k1}^{\frac{2}{\varepsilon_{k}}}\mathrm{e}^{-\lambda_{k}\frac{2\pi}{\varepsilon_{k}}\left(sP_{k}\right)^{\frac{2}{\varepsilon_{k}}}\mathbb{E}_{\Psi_{k}}\left[\Psi_{k}^{\frac{2}{\varepsilon_{k}}}\Gamma\left(-\frac{2}{\varepsilon_{k}},\frac{su\Psi_{k}}{\gamma_{k}\Psi_{k1}}\right)\right]}\right]}{u\mathrm{e}^{\frac{su}{\gamma_{k}}}}}{\int_{x=0}^{1}\frac{\mathbb{E}_{\Psi_{k1}}\left[\Psi_{k1}^{\frac{2}{\varepsilon_{k}}}\mathrm{e}^{-\lambda_{k}\frac{2\pi}{\varepsilon_{k}}\left(sP_{k}\right)^{\frac{2}{\varepsilon_{k}}}\mathbb{E}_{\Psi_{k}}\left[\Psi_{k}^{\frac{2}{\varepsilon_{k}}}\Gamma\left(-\frac{2}{\varepsilon_{k}},\frac{xux\Psi_{k}}{\gamma_{k}\Psi_{k1}}\right)\right]}\right]}{x^{\frac{2}{\varepsilon_{k}}+1}\mathrm{e}^{\frac{sux}{\gamma_{k}}}}dx}.\label{eq:derivativeLTnearestBS}
\end{eqnarray}
When fading coefficients are i.i.d. unit mean exponential random variables
$\mathbb{E}\left[\Psi_{k}^{\frac{2}{\varepsilon_{k}}}\right]$ $=\Gamma\left(1+\frac{2}{\varepsilon_{k}}\right)$,
setting $\left\{ \lambda_{cl}\right\} _{l=1}^{L}=0$ and $\left\{ \varepsilon_{k}\right\} _{k=1}^{K}=\alpha$,
(\ref{eq:derivativeLTmaxSINR}) reduces to \cite[(2)]{Dhillon2011b}.
Having computed the expressions for the joint p.d.f.'s in (\ref{eq:jointpdfmaxSINR})
and (\ref{eq:jointpdfnearestBS}), the coverage probabilities can
be easily obtained as shown below.
\begin{thm}
\label{thm:covProbExpression}The hetnet coverage probability max-SINR
connectivity and the nearest-BS connectivity models are as follows:
\begin{eqnarray}
 &  & \mathbb{P}_{\mathrm{coverage}}^{\mathrm{max-SINR}}=\sum_{i=1}^{K}\lambda_{i}\frac{2\pi}{\varepsilon_{i}}\left(\gamma_{i}P_{i}\right)^{\frac{2}{\varepsilon_{i}}}\mathbb{E}\left[\Psi_{k}^{\frac{2}{\varepsilon_{k}}}\right]\times\nonumber \\
 &  & \int_{y=0}^{\infty}\int_{\omega=-\infty}^{\infty}\mathcal{L}_{I_{o}+I_{c}+\eta,\ \underset{i=1,\cdots,K}{\max}\gamma_{i}M_{i}\le y}\left(j\omega\right)\times\nonumber \\
 &  & \frac{\left(\mathrm{e}^{j\omega y\left(1-\gamma_{i}^{-1}\right)}-\mathrm{e}^{j\omega\left(\eta+y\left(\kappa^{-1}-\gamma_{i}^{-1}\right)\right)}\right)}{2\pi j\omega y^{1+\frac{2}{\varepsilon_{i}}}}d\omega dy,\label{eq:covMaxSINR}\\
 &  & \mathbb{P}_{\mathrm{coverage}}^{\mathrm{nearest}}=\nonumber \\
 &  & \int_{y=0}^{\infty}\int_{\omega=-\infty}^{\infty}\left.\frac{\partial}{\partial u}\mathcal{L}_{I_{o}+I_{c}+\eta,\ \underset{i=1,\cdots,K}{\max}\gamma_{i}N_{i}\le u}\left(j\omega\right)\right|_{u=y}\times\nonumber \\
 &  & \frac{\mathrm{e}^{j\omega y}-\mathrm{e}^{j\omega\left(\frac{y}{\kappa}+\eta\right)}}{j\omega2\pi}d\omega dy,\label{eq:covNearestBS}
\end{eqnarray}
where $\kappa=\underset{i=1,\cdots,\ K}{\max}\gamma_{i}$, all the
other symbols are in Table \ref{tab:symbolList}, and the Laplace
transform function in (\ref{eq:covMaxSINR}) and the derivative of
the Laplace transform function in (\ref{eq:covNearestBS}) are given
in (\ref{eq:maxSINRLTexp}) and (\ref{eq:derivativeLTnearestBS}),
respectively.\end{thm}
\begin{proof}
Once the joint p.d.f. has been obtained (see (\ref{eq:jointpdfmaxSINR})
and (\ref{eq:jointpdfnearestBS})), the probability of the event in
(\ref{eq:maxSINRcoverageEqexpr}) can be derived as follows:
\begin{eqnarray*}
 &  & \mathbb{P}_{coverage}^{\mathrm{max-SINR}}\\
 &  & =\mathbb{P}\left(\left\{ \frac{1}{\kappa}\times\underset{i=1,\cdots,K}{\max}\gamma_{i}M_{i}+\eta<I<\underset{i=1,\cdots,K}{\max}\gamma_{i}M_{i}\right\} \right)\\
 &  & \overset{\left(a\right)}{=}\int_{y=0}^{\infty}\int_{x=\frac{y}{\kappa}+\eta}^{y}f_{I,\underset{i=1,\cdots,K}{\max}\gamma_{i}M_{i}}\left(x,y\right)dxdy\\
 &  & \overset{\left(b\right)}{=}\int_{y=0}^{\infty}\int_{\omega=-\infty}^{\infty}\left.\frac{\partial}{\partial u}\mathcal{L}_{I_{o}+I_{c}+\eta,\ \underset{i=1,\cdots,K}{\max}\gamma_{i}M_{i}\le y}\left(j\omega\right)\right|_{u=y}\times\\
 &  & \frac{\mathrm{e}^{j\omega y}-\mathrm{e}^{j\omega\left(\frac{y}{\kappa}+\eta\right)}}{j\omega2\pi}d\omega dy,
\end{eqnarray*}
where $\left(a\right)$ expresses the probability of the coverage
event in terms of the joint p.d.f., $\left(b\right)$ is obtained
by substituting for the joint p.d.f. from (\ref{eq:jointpdfmaxSINR}),
then interchanging the order of integrations of the variables $x$
and $\omega$ which is justified by the boundedness of the integrals.
Finally, the above expression can be further simplified to obtain
(\ref{eq:covMaxSINR}). 

The same steps can be followed for obtaining (\ref{eq:covNearestBS}),
and are omitted for brevity.
\end{proof}
Using an alternate approach, expressions for the hetnet coverage probability
are obtained in \cite{Mukherjee2012a}, again, when all the fading
coefficients are i.i.d. exponential random variables. For a general
system model as in this paper, to the best of our knowledge, the hetnet
coverage probability has not been characterized, until now.

Nevertheless, the semi-analytical expressions are extremely complicated
even for numerical computations, and little intuition and insights
about the hetnet performances are obtainable from these expressions.
As a result, a more qualitative study is imperative to better understand
these soon-to-be-prevalent cellular networks. From now onwards, we
conduct a more systematic study to bring out the properties and dependencies
of the hetnet performance on the various parameters of the system.
To begin with, we make the following observations about the hetnet
performance. 
\begin{cor}
\label{thm:HetnetEquivalence}The downlink coverage probability in
the hetnet is the same as in another hetnet with the same open-access
tiers as in the original hetnet (described in Section \ref{sub:BS-Layout})
and one closed access tier where the BSs have unity transmission power,
fading coefficient and path-loss exponent and are arranged according
to a non-homogeneous Poisson point process with a BS density function
\begin{eqnarray}
\lambda_{c}\left(r\right) & = & \sum_{l=1}^{L}\lambda_{cl}P_{cl}^{\frac{2}{\varepsilon_{cl}}}\mathbb{E}\left[\Psi_{cl}^{\frac{2}{\varepsilon_{cl}}}\right]r^{\frac{2}{\varepsilon_{cl}}-1},\ r\ge0.\label{eq:closedAccessEqBSDensityFunction}
\end{eqnarray}
\end{cor}
\begin{proof}
Firstly, the BSs in the closed-access tiers only contribute to the
interference as the MSs cannot be served by these BSs. Secondly, the
total closed-access interference power $I_{c}$ is independent of
the signal power and interference power at the MS from the open-access
tiers. Next, $I_{c}$ satisfies the following stochastic equivalence
$I_{c}=_{\mathrm{st}}\sum_{n=1}^{\infty}\tilde{R}_{n}^{-1},$ where
$\left\{ \tilde{R}_{n}\right\} _{n=1}^{\infty}$ is the set of distances
from the origin of BSs arranged according to a non-homogeneous Poisson
point process with BS density function given in (\ref{eq:closedAccessEqBSDensityFunction}).
This is obtained by first using the \cite[Theorem 2]{Madhusudhanan2010a}
to obtain an equivalent BS arrangement for each closed-access tier
according to non-homogeneous Poisson point process with unity transmission
power, fading coefficient and path-loss exponent at each BS in the
tier. Due to \cite[Theorem 2]{Madhusudhanan2010a}, the equivalent
BS arrangement has the same probability distribution for the interference
caused at the MS as the original case. Next, since the BS arrangements,
transmission and fading characteristics of the BSs of all closed-access
tiers are independent of each other, using the Superposition theorem
\cite[Page 16]{Kingman1993}, the $L$ closed-access tiers can be
combined together to obtain a single closed-access tier with BS according
to non-homogeneous Poisson point process with a BS density function
equal to the sum of the BS density functions of the individual tiers
obtained from the previous step, and is shown in (\ref{eq:closedAccessEqBSDensityFunction}).
Again, the equivalence is such that the probability distribution of
$I_{c}$ will be the same as that of the equivalent closed-access
tier where the BSs have unity tranmission power and fading coefficients.
\end{proof}

Hence, we have shown an equivalence between a hetnet with $L$ closed-access
tiers and another hetnet with a single closed-access tier. Next, we
make an interesting observation regarding the hetnet downlink performance
under the max-SINR connectivity model.
\begin{cor}
\label{cor:arbitraryFadingCorollary} The hetnet performance under
max-SINR connectivity with an arbitrary fading distribution at each
tier is the same as in another hetnet with open-access and closed-access
BS densities as $\left\{ \left.\lambda_{oi}\mathbb{E}\Psi_{oi}^{\frac{2}{\varepsilon_{i}}}\right/\Gamma\left(1+\frac{2}{\varepsilon_{i}}\right)\right\} _{i=1}^{K}$
and $\left\{ \left.\lambda_{ci}\mathbb{E}\Psi_{ci}^{\frac{2}{\varepsilon_{ci}}}\right/\Gamma\left(1+\frac{2}{\varepsilon_{ci}}\right)\right\} _{i=1}^{L}$,
respectively, and i.i.d. unit mean exponential distribution for fading
at all the BSs in the network.
\end{cor}

The above result is obtained by noting that the effect of fading is
equivalent to scaling the density of BSs by the $\frac{2}{\varepsilon}^{\mathrm{th}}$
moment of the fading random variable, due to \cite[Corollary 2]{Madhusudhanan2010a}.
A large body of work involving the stochastic geometric study of networks
predominantly assume fading coefficients to be i.i.d. exponential
random variables, as this greatly simplifies the analysis and renders
itself to closed-form characterization of coverage probabilities and
other related performance metrics of several networks including the
hetnets (see \cite{Dhillon2012a}). A common criticism for all these
works has been that the exponential distribution does not accurately
capture the slow fading environment. Interestingly, the above corollary
shows an example of a scenario wherein studies with exponential fading
assumptions completely characterizes the arbitrary fading scenario.
Unfortunately, the same is not true for the nearest-BS connectivity
model. In the following section, we explore more properties for the
hetnet downlink performance.

The importance of the Corollary \ref{thm:HetnetEquivalence} is that
the SINR distribution of the two equivalent hetnets are the same for
both the max-SINR and nearest-BS connectivity models. Hence, without
loss of generality, we study the downlink performance where the hetnet
consists of $K$ tiers of open access networks and a single closed
access network. For the sake of simplicity, it is assumed that the
closed-access tier has homogeneous Poisson point process based BS
arrangement with a constant BS density $\lambda_{c}$, transmission
power $P_{c}$, path-loss exponent $\varepsilon_{c}$ and i.i.d. fading
coefficients with the same distribution as $\Psi_{c}\ \left(\mathbb{E}\left[\Psi_{c}^{\frac{2}{\varepsilon_{c}}}\right]<\infty\right)$.

\section{Qualitative study of hetnet downlink performance}

We begin with some simple stochastic ordering results comparing the
hetnet coverage probabilities for the two connectivity models. 
\begin{prop}
\label{pro:stOrderingmaxSINRNearest}For the same system parameters,
$\mathbb{P}_{\mathrm{coverage}}^{\mathrm{max-SINR}}>\mathbb{P}_{\mathrm{coverage}}^{\mathrm{nearest}}$.
\end{prop}
The above result is easily proved by noting from (\ref{eq:maxSINRcoverageDef})
and (\ref{eq:nearestBScoverageDef}) that $\bigcup_{k=1}^{K}\bigcup_{i=1}^{\infty}\left\{ \mathrm{SINR}_{ki}>\beta_{k}\right\} \supset\bigcup_{k=1}^{K}\left\{ \mathrm{SINR}_{k1}>\beta_{k}\right\} $,
i.e. the coverage event corresponding to the nearest-BS connectivity
is a subset of the max-SINR connectivity model. 

When $\left\{ \beta_{k}\right\} _{k=1}^{\infty}=\beta$, commonly
referred to as the unbiased case in the literature, the hetnet coverage
probabilities of the max-SINR connectivity model is identical to the
maximum instantaneous received power (MIRP) connectivity model. Under
the MIRP connectivity, the MS connects to the BS with the maximum
instantaneous received power among all the open-access tiers. As a
result, the serving BS and the coverage probability expression for
the MIRP are 
\begin{eqnarray}
\left(T,I\right) & = & \underset{k=1,\cdots,\ K,\ i=1,\ 2,\cdots}{\mathrm{argmax}}P_{ok}\Psi_{ki}R_{ki}^{-\varepsilon_{k}},\nonumber \\
\mathbb{P}_{\mathrm{coverage}}^{\mathrm{MIRP}} & = & \mathbb{P}\left(\left\{ \mathrm{SINR}_{T,I}>\beta_{T}\right\} \right),\label{eq:tierExpCoverageExpMIRP}
\end{eqnarray}
where $T$ refers to the tier-index and $I$ refers to the BS-index
of the serving BS. Another popular connectivity model for the hetnets
is the so-called maximum biased-received-power (MBRP) connectivity
model that is studied in \cite{Jo2011}. 

Under MBRP, the MS associates with the BS with the maximum average
received-power in the hetnet with a certain biasing corresponding
to each tier. Hence, the serving BS will be one of the nearest BSs
from the MS corresponding to the $K$ open-access tiers. Then, the
tier-index of the serving BS and the hetnet coverage probability under
MBRP are determined as follows: 
\begin{eqnarray}
T & = & \underset{k=1,\cdots,\ K}{\mathrm{argmax}}\ \underset{i=1,\ 2,\cdots}{\max}P_{ok}\mathbb{E}\left[\Psi_{ki}\right]R_{ki}^{-\varepsilon_{k}}B_{ok},\nonumber \\
 & = & \underset{k=1,\cdots,\ K}{\mathrm{argmax}}P_{ok}\mathbb{E}\left[\Psi_{k}\right]R_{k1}^{-\varepsilon_{k}}B_{ok}\label{eq:tierExpMBRP}\\
\mathbb{P}_{\mathrm{coverage}}^{\mathrm{MBRP}} & = & \mathbb{P}\left(\left\{ \mathrm{SINR}_{T,1}>\beta_{T}\right\} \right),\label{eq:CoverageExpMBRP}
\end{eqnarray}
where $\left\{ B_{ok}(>0)\right\} _{k=1}^{K}$ are the biasing factors;
$\left\{ P_{k}\mathbb{E}\left[\Psi_{ki}\right]R_{ki}^{-\varepsilon_{k}}\right\} _{i=1}^{\infty}$
is the long-term averaged received power at the MS from the $k^{\mathrm{th}}$
tier BSs, and the maximum is from the $k^{\mathrm{th}}$ tier BS nearest
to the MS; and $\mathrm{SINR}_{k,i}$ is defined in $\left(\ref{eq:SINRDefinition}\right)$. 

We begin with characterizing the c.c.d.f. of $\mathrm{SINR}_{T,I}$
for the MIRP case, and several related important characteristics.

\subsection{\label{sub:MIRP}SINR characterization under MIRP connectivity}

The following stochastic equivalence helps simplify the SINR characterization. 
\begin{lem}
\textup{\label{lem:StochasticEqLemma1}}The SINR at the MS under MIRP
is the same as in the two-tier hetnet where the tier to which the
MS has an open-access network with a BS density function $\tilde{\lambda}\left(r\right)=\sum_{k=1}^{K}\tilde{\lambda}_{k}\left(r\right)$
with \textup{$\tilde{\lambda}_{k}\left(r\right)=\lambda_{k}\frac{2\pi}{\varepsilon_{k}}P_{k}^{\frac{2}{\varepsilon_{k}}}\mathbb{E}\Psi_{k}^{\frac{2}{\varepsilon_{k}}}r^{\frac{2}{\varepsilon_{k}}-1},\ r\ge0$}
and a closed-access network with a BS density function $\hat{\lambda}\left(r\right)=\lambda_{c}\frac{2\pi}{\varepsilon_{c}}P_{c}^{\frac{2}{\varepsilon_{c}}}\mathbb{E}\left[\Psi_{c}^{\frac{2}{\varepsilon_{c}}}\right]r^{\frac{2}{\varepsilon_{c}}-1}$.
All the BSs in the equivalent systems have unity transmit powers,
fading coefficients and path-loss exponents. The SINR is stochastically
equal to 
\begin{eqnarray}
 &  & \mathrm{SINR}_{T,I}\nonumber \\
 &  & =_{\mathrm{st}}\left.\frac{\tilde{R}_{T,1}^{-1}}{\underset{\left(k,l\right)\ne\left(T,1\right)}{\sum_{k=1}^{K}\sum_{l=1}^{\infty}}\tilde{R}_{kl}^{-1}+\sum_{l=1}^{\infty}\hat{R}_{l}^{-1}+\eta}\right|_{\left(\left\{ \tilde{\lambda}_{k}\left(r\right)\right\} _{k=1}^{K},\hat{\lambda}\left(r\right)\right)}\nonumber \\
 &  & =_{\mathrm{st}}\left.\frac{\tilde{R}_{1}^{-1}}{\sum_{k=2}^{\infty}\tilde{R}_{k}^{-1}+\sum_{l=1}^{\infty}\hat{R}_{l}^{-1}+\eta}\right|_{\left(\tilde{\lambda}\left(r\right),\hat{\lambda}\left(r\right)\right)},\label{eq:SINRDistributionEq1}
\end{eqnarray}
where $=_{\mathrm{st}}$ indicates the equivalence in distribution;
and $\left\{ \tilde{R}_{i}\right\} _{i=1}^{\infty}$ $\left(\left\{ \hat{R}_{i}\right\} _{i=1}^{\infty}\right)$
is the ascendingly ordered distances of the BSs from the origin, obtained
from a non-homogeneous 1-D Poisson point process with BS density function
$\tilde{\lambda}\left(r\right)$ $\left(\hat{\lambda}\left(r\right)\right)$
defined above.\end{lem}
\begin{proof}
See Appendix \ref{sub:proofStEquivalenceLemma1}.
\end{proof}
The following lemma shows interesting stochastic equivalences when
$\left\{ \varepsilon_{k}\right\} _{k=1}^{K}=\varepsilon_{c}=\varepsilon$.
\begin{lem}
\label{lem:StochasticEqLemma2}The hetnet SINR under MIRP connectivity
has the same distribution as in the following three networks. The
first is a hetnet with BS densities $\left\{ \lambda_{k}P_{k}^{\frac{2}{\varepsilon}}\mathbb{E}\left[\Psi_{k}^{\frac{2}{\varepsilon}}\right]\right\} _{k=1}^{K},\ \lambda_{c}P_{k}^{\frac{2}{\varepsilon}}\mathbb{E}\left[\Psi_{c}^{\frac{2}{\varepsilon}}\right]$
for the $K$ open-access tiers and the closed-access tier, respectively,
unity transmit powers and shadow fading factors for all tiers. The
other two are two-tier networks with unity transmit powers and shadow
fading factors for all their BSs. The first two-tier network has the
open-access tier BS density $\sum_{l=1}^{K}\lambda_{l}P_{l}^{\frac{2}{\varepsilon}}\mathbb{E}\left[\Psi_{l}^{\frac{2}{\varepsilon}}\right]$,
closed-access tier BS density $\lambda_{c}P_{k}^{\frac{2}{\varepsilon}}\mathbb{E}\left[\Psi_{c}^{\frac{2}{\varepsilon}}\right]$and
experiences the same background noise as the hetnets. The second two-tier
network has a unity open-access tier BS density, closed-access tier
BS density $\hat{\lambda}_{c}=\frac{\lambda_{c}P_{k}^{\frac{2}{\varepsilon}}\mathbb{E}\left[\Psi_{c}^{\frac{2}{\varepsilon}}\right]}{\sum_{l=1}^{K}\lambda_{l}P_{l}^{\frac{2}{\varepsilon}}\mathbb{E}\left[\Psi_{l}^{\frac{2}{\varepsilon}}\right]}$
and a background noise $\bar{\eta}=\eta\left(\sum_{l=1}^{K}\lambda_{l}P_{l}^{\frac{2}{\varepsilon}}\mathbb{E}\left[\Psi_{l}^{\frac{2}{\varepsilon}}\right]\right)^{-\frac{\varepsilon}{2}}$.
Equivalently, 
\begin{eqnarray}
 &  & \hspace{-0.75cm}\mathrm{SINR}_{T,I}=_{\mathrm{st}}\nonumber \\
 &  & \hspace{-0.75cm}\mathrm{SINR}\left(K+1,\left\{ \lambda_{k}P_{k}^{\frac{2}{\varepsilon}}\mathbb{E}\left[\Psi_{k}^{\frac{2}{\varepsilon}}\right]\right\} _{k=1}^{K},\lambda_{c}P_{k}^{\frac{2}{\varepsilon}}\mathbb{E}\left[\Psi_{c}^{\frac{2}{\varepsilon}}\right],\eta,T\right)\label{eq:SINRDistributionEq2}\\
 &  & \hspace{-0.75cm}=_{\mathrm{st}}\mathrm{SINR}\left(2,\sum_{l=1}^{K}\lambda_{l}P_{l}^{\frac{2}{\varepsilon}}\mathbb{E}\left[\Psi_{l}^{\frac{2}{\varepsilon}}\right],\lambda_{c}P_{c}^{\frac{2}{\varepsilon}}\mathbb{E}\left[\Psi_{c}^{\frac{2}{\varepsilon}}\right],\eta,1\right)\label{eq:SINRDistEqDensity1System2}\\
 &  & \hspace{-0.75cm}=_{\mathrm{st}}\mathrm{SINR}\left(2,1,\hat{\lambda}_{c},\bar{\eta},1\right),\label{eq:SINRDistEqDensity1System}
\end{eqnarray}
 where $=_{\mathrm{st}}$ indicates equivalence in distribution. The
SINR expression on the right-hand side is a function of the total
number of tiers in the hetnet, BS densities of each open-access tier,
BS densities of the closed-access tier, back-ground noise power, and
the tier index of the serving BS, respectively.\end{lem}
\begin{proof}
See Appendix \ref{sub:proofStEquivalenceLemma2}.
\end{proof}

Lemmas \ref{lem:StochasticEqLemma1} and \ref{lem:StochasticEqLemma2}
are generalizations of \cite[Lemma 1]{Madhusudhanan2012} and \cite[Lemma 1]{Madhusudhanan2012a},
respectively, to the case where the hetnet also contains a closed-access
tier. Next, we compute the hetnet coverage probability. 
\begin{thm}
\label{thm:covprobMIRPexp}The hetnet coverage probability under MIRP
is 
\begin{eqnarray}
 &  & \mathbb{P}_{\mathrm{coverage}}^{\mathrm{MIRP}}=\nonumber \\
 &  & \sum_{k=1}^{K}\lambda_{k}P_{k}^{\frac{2}{\varepsilon_{k}}}\mathbb{E}\left[\Psi_{k}^{\frac{2}{\varepsilon_{k}}}\right]\times\nonumber \\
 &  & \int_{r=0}^{\infty}2\pi r\int_{\omega=-\infty}^{\infty}\frac{\mathrm{e}^{j\omega\eta r^{\varepsilon_{k}}}\left(1-\mathrm{e}^{-\frac{j\omega}{\beta k}}\right)}{j\omega2\pi}\times\nonumber \\
 &  & \mathrm{e}^{-\lambda_{c}P_{c}^{\frac{2}{\varepsilon_{c}}}\mathbb{E}\left[\Psi_{c}^{\frac{2}{\varepsilon_{c}}}\right]\pi r^{\frac{2\varepsilon_{k}}{\varepsilon_{c}}}G\left(j\omega,\frac{2}{\varepsilon_{c}}\right)}\times\nonumber \\
 &  & \mathrm{e}^{-\sum_{l=1}^{K}\lambda_{l}P_{l}^{\frac{2}{\varepsilon_{l}}}\mathbb{E}\left[\Psi_{l}^{\frac{2}{\varepsilon_{l}}}\right]\pi r^{\frac{2\varepsilon_{k}}{\varepsilon_{l}}}\ _{1}F_{1}\left(-\frac{2}{\varepsilon_{l}};1-\frac{2}{\varepsilon_{l}};j\omega\right)}d\omega dy,\label{eq:covProbMIRPAlternateExp}
\end{eqnarray}
 where $G\left(j\omega,\frac{2}{\varepsilon_{c}}\right)=\int_{t=0}^{\infty}\left(1-\mathrm{e}^{j\omega t}\right)\frac{2}{\varepsilon_{c}}t^{-1-\frac{2}{\varepsilon_{c}}}dt$.\end{thm}
\begin{proof}
The proof is along the same lines as \cite[Theorem 1]{Madhusudhanan2012},
and is not shown here.
\end{proof}
The above expression can be greatly simplified under certain special
cases, and the following results present these cases.
\begin{cor}
\label{cor:covprobMIRPsameBeta}When $\left\{ \varepsilon_{k}\right\} _{k=1}^{K}=\varepsilon_{c}=\varepsilon$,
the hetnet coverage probability is 
\begin{eqnarray}
 &  & \mathbb{P}_{\mathrm{coverage}}^{\mathrm{MIRP}}=\nonumber \\
 &  & \sum_{k=1}^{K}\frac{\lambda_{k}P_{k}^{\frac{2}{\varepsilon}}\mathbb{E}\left[\Psi_{k}^{\frac{2}{\varepsilon}}\right]\int_{\omega=-\infty}^{\infty}\frac{\left(1-\mathrm{e}^{-\frac{j\omega}{\beta_{k}}}\right)}{j\omega2\pi}H\left(j\omega\right)d\omega}{\sum_{l=1}^{K}\lambda_{l}P_{l}^{\frac{2}{\varepsilon}}\mathbb{E}\left[\Psi_{l}^{\frac{2}{\varepsilon}}\right]},\label{eq:covProbExpressionMIRPConnectivitySpCase1}\\
 &  & H\left(j\omega\right)=\int_{r=0}^{\infty}2\pi r\times\nonumber \\
 &  & \mathrm{e}^{j\omega\bar{\eta}r^{\varepsilon}-\pi r^{2}\left(\ _{1}F_{1}\left(-\frac{2}{\varepsilon};1-\frac{2}{\varepsilon};j\omega\right)+\hat{\lambda}_{c}G\left(j\omega,\frac{2}{\varepsilon}\right)\right)}dr
\end{eqnarray}
where $\left.H\left(j\omega\right)\right|_{\bar{\eta}=0}=\frac{1}{\ _{1}F_{1}\left(-\frac{2}{\varepsilon};1-\frac{2}{\varepsilon};j\omega\right)+\hat{\lambda}_{c}G\left(j\omega,\frac{2}{\varepsilon}\right)}$,
$\bar{\eta}$ and $\hat{\lambda}_{c}$ are from Lemma \ref{lem:StochasticEqLemma2}
and $G\left(\cdot,\cdot\right)$ is defined in Theorem \ref{thm:covprobMIRPexp}.
When $\left\{ \beta_{k}\right\} _{k=1}^{K}=\beta$ or $\left\{ \beta_{k}\right\} _{k=1}^{K}\ge1$,
(\ref{eq:covProbExpressionMIRPConnectivitySpCase1}) is equal to $\mathbb{P}_{\mathrm{coverage}}^{\mathrm{max-SINR}}$.
When there is no closed-access tier $\left(\hat{\lambda}_{c}=0\right)$,
(\ref{eq:covProbExpressionMIRPConnectivitySpCase1}) is equal to the
single-tier network coverage probability (see \cite[Corollary 4]{Madhusudhanan2010a})
and is independent of the transmission powers and fading factors of
the BSs in the system.\end{cor}
\begin{proof}
The result is obtained by exchanging the order of integrations in
(\ref{eq:covProbMIRPAlternateExp}) and simplifying. 
\end{proof}
The following theorem shows another scenario when the hetnet coverage
probabilities are identical for the max-SINR and MIRP connectivity
models. 
\begin{thm}
\label{thm:covprobMIRPBetaGt1}When $\beta_{k}\ge1,\ \forall\ k=1,\cdots,K$,
the hetnet coverage probability is given by 
\begin{eqnarray}
 &  & \mathbb{P}_{\mathrm{coverage}}^{\mathrm{max-SINR}}=\mathbb{P}_{\mathrm{coverage}}^{\mathrm{MIRP}}=\sum_{k=1}^{K}\frac{\lambda_{ok}P_{ok}^{\frac{2}{\varepsilon_{k}}}\mathbb{E}\left[\Psi_{k}^{\frac{2}{\varepsilon_{k}}}\right]\beta_{k}^{-\varepsilon_{k}}}{\Gamma\left(1+\frac{2}{\varepsilon_{k}}\right)}\times\nonumber \\
 &  & \int_{r=0}^{\infty}2\pi r\times\mathrm{e}^{-\eta r^{\varepsilon_{k}}-\frac{\lambda_{c}\pi P_{c}^{\frac{2}{\varepsilon_{c}}}\mathbb{E}\left[\Psi_{c}^{\frac{2}{\varepsilon_{c}}}\right]r^{\frac{2\varepsilon_{k}}{\varepsilon_{c}}}}{\Gamma\left(1+\frac{2}{\varepsilon_{c}}\right)\mathrm{sinc}\left(\frac{2\pi}{\varepsilon_{c}}\right)}}\nonumber \\
 &  & \mathrm{e}^{-\sum_{l=1}^{K}\frac{\lambda_{ol}\pi P_{ol}^{\frac{2}{\varepsilon_{l}}}\mathbb{E}\left[\Psi_{l}^{\frac{2}{\varepsilon_{l}}}\right]r^{\frac{2\varepsilon_{k}}{\varepsilon_{l}}}}{\Gamma\left(1+\frac{2}{\varepsilon_{l}}\right)\mathrm{sinc}\left(\frac{2\pi}{\varepsilon_{l}}\right)}}dr,\label{eq:covProbExpressionBetaGt1}
\end{eqnarray}
and in the interference limited case $\left(\eta=0\right)$ when $\left\{ \varepsilon_{k}\right\} _{k=1}^{K}=\varepsilon_{c}=\varepsilon$
\begin{eqnarray}
 &  & \mathbb{P}_{\mathrm{coverage}}^{\mathrm{max-SINR}}=\mathbb{P}_{\mathrm{coverage}}^{\mathrm{MIRP}}\nonumber \\
 &  & =\sum_{k=1}^{K}\frac{\lambda_{ok}P_{ok}^{\frac{2}{\varepsilon}}\mathbb{E}\left[\Psi_{k}^{\frac{2}{\varepsilon}}\right]\mathrm{sinc}\left(\frac{2\pi}{\varepsilon}\right)\beta_{k}^{-\varepsilon}}{\lambda_{c}P_{c}^{\frac{2}{\varepsilon}}\mathbb{E}\left[\Psi_{c}^{\frac{2}{\varepsilon}}\right]+\sum_{l=1}^{K}\lambda_{ol}P_{ol}^{\frac{2}{\varepsilon}}\mathbb{E}\left[\Psi_{l}^{\frac{2}{\varepsilon}}\right]}.\label{eq:covprobExpBetaGt1NoNoise}
\end{eqnarray}
\end{thm}
\begin{proof}
Firstly, from \cite[Lemma 1]{Dhillon2012a}, when $\beta_{k}\ge1$,
there exists at most one open-access BS that can have an SINR above
the corresponding threshold. As a result, hetnet coverage probability
in (\ref{eq:maxSINRcoverageDef}) becomes $\mathbb{P}_{\mathrm{coverage}}^{\mathrm{max-SINR}}=\sum_{k=1}^{K}\mathbb{P}\left(\left\{ \mathrm{SINR}_{k}\left(\max\right)>\beta_{k}\right\} \right)=\mathbb{P}_{\mathrm{coverage}}^{\mathrm{MIRP}}$.
See Appendix \ref{sub:proofCovprobMIRPBetaGt1} to derive (\ref{eq:covProbExpressionBetaGt1}),
which simplifies to (\ref{eq:covprobExpBetaGt1NoNoise}) when $\eta=0$.
\end{proof}

In the above result, (\ref{eq:covProbExpressionBetaGt1}) can be easily
computed numerically and is an extension of \cite[Theorem 1]{Dhillon2012a}
to arbitrary fading and path-loss case. The study of the MIRP connectivity
has given many interesting insights and simplifications for the max-SINR
case. 

Further, other performance metrics pertinent to hetnets such as the
average fraction of load carried by each tier in the hetnet and the
area-averaged rate acheived by an MS that is in coverage in a hetnet
can also be derived using the results in this section. We refer the
reader to \cite[Theorems 2, 3 and 4]{Madhusudhanan2012c} for these
results. 

Now, we study the MBRP connectivity in further detail, and derive
interesting results and relationships with the hetnet performance
under nearest-BS connectivity.

\subsection{\label{sub:MARP}SINR characterization under MBRP connectivity}

From the definition of the hetnet coverage probability under MIRP
and MBRP, the stochastic ordering result can be extended beyond Proposition
\ref{pro:stOrderingmaxSINRNearest} as follows.
\begin{prop}
For the same system parameters, when $\left\{ \beta_{k}\right\} _{k=1}^{K}=\beta$
or $\left\{ \beta_{k}\right\} _{k=1}^{K}\ge1$, $\mathbb{P}_{\mathrm{coverage}}^{\mathrm{MBRP}}=\mathbb{P}_{\mathrm{coverage}}^{\mathrm{nearest}}<\mathbb{P}_{\mathrm{coverage}}^{\mathrm{max-SINR}}=\mathbb{P}_{\mathrm{coverage}}^{\mathrm{MIRP}}$,
for the biasing factors under MBRP as $B_{ok}=\frac{1}{P_{ok}\mathbb{E}\left[\Psi_{k}\right]}$,
for $k=1,\dots,\ K$ open-access tiers. When $\left\{ B_{ok}\right\} _{k=1}^{K}=1$,
MBRP connectivity is same as the popular maximum averaged received
power (MARP) connectivity model. 
\end{prop}

It is clear from equations (\ref{eq:nearestBScoverageEqexpr}) and
(\ref{eq:derivativeLTnearestBS}) that it is tedious to compute the
hetnet coverage probability under nearest-BS connectivity, even with
numerical integration, for arbitrary fading case. With slight modifications
to the approach in Theorem \ref{thm:covProbExpression} and \cite[Theorem 1]{Nguyen2010},
hetnet coverage probability with MBRP can also be derived. These expressions
do not simplify significantly beyond that in (\ref{eq:nearestBScoverageEqexpr})
and hence is not presented here. Hence, we conduct a similar qualitative
study of the hetnet performance under MBRP, as in Section \ref{sub:MIRP}.
\begin{cor}
\label{thm:stEquivalenceMARP}Under MBRP connectivity, the following
stochastic equivalences hold:
\begin{eqnarray}
 &  & \hspace{-1cm}\mathrm{SINR}_{T,1}=_{\mathrm{st}}\nonumber \\
 &  & \hspace{-1cm}\left.\frac{\tilde{P}_{T}\Psi_{T,1}\tilde{R}_{T,1}^{-1}}{\underset{\left(m,l\right)\ne\left(T,1\right)}{\sum_{m=1}^{K}\sum_{l=1}^{\infty}}\tilde{P}_{m}\Psi_{ml}\tilde{R}_{ml}^{-1}+\sum_{n=1}^{\infty}\Psi_{cn}\hat{R}_{cn}^{-1}+\eta}\right|_{\left(\left\{ \tilde{\lambda}_{k}\left(r\right)\right\} _{k=1}^{K},\hat{\lambda}\left(r\right)\right)},\label{eq:stEquivalence MARP2}
\end{eqnarray}

where the equivalent hetnet has BS distributions according to non-homogeneous
Poisson process with density functions $\left\{ \tilde{\lambda}_{k}\left(r\right)=\lambda_{ok}\frac{2\pi}{\varepsilon_{k}}\left(P_{ok}\mathbb{E}\left[\Psi_{k}\right]B_{ok}\right)^{\frac{2}{\varepsilon_{k}}}r^{\frac{2}{\varepsilon_{k}}-1}\right\} _{k=1}^{K}$,
$\hat{\lambda}\left(r\right)=\lambda_{c}\frac{2\pi}{\varepsilon_{c}}P_{c}^{\frac{2}{\varepsilon_{c}}}r^{\frac{2}{\varepsilon_{c}}-1}$,
$r\ge0$, for the $K$ open-access tiers and the closed-access tier,
respectively, , transmission power of the $m^{\mathrm{th}}\ \left(m=1,\dots,\ K\right)$
open-access tier BSs as $\tilde{P}_{m}=\left(\mathbb{E}\left[\Psi_{l}\right]B_{ol}\right)^{-1}$
and unity for the closed-access tier BSs; and unity path-loss exponent
for all BSs, but the fading distributions are the same as the original
hetnet.\end{cor}
\begin{proof}
\begin{eqnarray*}
 &  & \mathbb{P}\left(\left\{ \mathrm{SINR}_{T,1}>\beta\right\} \right)\\
 &  & =\mathbb{P}\left(\left\{ P_{oT}\Psi_{T1}R_{T1}^{-\varepsilon_{T}}\left/\left(\underset{\left(m,n\right)\ne\left(T,1\right)}{\sum_{m=1}^{K}\sum_{n=1}^{\infty}}P_{om}\Psi_{mn}R_{mn}^{-\varepsilon_{m}}+\right.\right.\right.\right.\\
 &  & \left.\left.+\sum_{l=1}^{L}\sum_{n=1}^{\infty}P_{cl}\Psi_{cln}R_{cln}^{-\varepsilon_{cl}}+\eta\right)>\beta\right\} \bigcap\\
 &  & \left.\left\{ \begin{array}{c}
P_{oT}R_{T1}^{-\varepsilon_{T}}\mathbb{E}\left[\Psi_{T}\right]B_{oT}>P_{om}R_{m1}^{-\varepsilon_{m}}\mathbb{E}\left[\Psi_{m}\right]B_{om},\\
m=1,\cdots,K,\ \left(m,n\right)\ne\left(T,1\right)
\end{array}\right\} \right),
\end{eqnarray*}
where the serving BS's tier-index is determined by the second event
in the above expression (see (\ref{eq:tierExpMBRP})). Now, as in
the proof of Lemma \ref{lem:StochasticEqLemma1}, using \cite[Theorem 2]{Madhusudhanan2012},
let us consider the set $\left\{ \tilde{R}_{mi}=\left(P_{om}\mathbb{E}\left[\Psi_{m}\right]B_{om}\right)^{-1}R_{mi}^{\varepsilon_{m}}\right\} _{i=1}^{\infty}$
as a set of distances from the origin of $m^{\mathrm{th}}$ tier BSs
from the origin in an equivalent hetnet. Using the Marking theorem
of Poisson process \cite[Page 55]{Kingman1993}, the equivalent hetnet
is from a non-homogeneous Poisson point process with a BS density
function $\tilde{\lambda}_{m}\left(r\right)$, $r\ge0$, $\forall\ m=1,\cdots,\ K$
as shown in Corollary \ref{thm:stEquivalenceMARP}. Further, the serving
BS is the closest among all the open-access BSs of the hetnet. Finally,
the $m^{\mathrm{th}}$ tier transmit power of the equivalent system,
$\tilde{P}_{m}$, is obtained to ensure that the received power at
the MS is stochastically equivalent to that of the original cellular
system.
\end{proof}
The result is yet another application of the Marking theorem of Poisson
process, and can be proved using the same techniques as developed
in Lemma \ref{lem:StochasticEqLemma1}. The following result shows
important characteristics of the serving BS under the MBRP case. 
\begin{lem}
The probability mass function (p.m.f.) of the serving BS's tier-index
and the joint p.d.f. of the serving BS's tier-index and distance from
the MS under MBRP connectivity are

\begin{eqnarray}
 &  & \mathbb{P}\left(\left\{ T=k\right\} \right)=\nonumber \\
 &  & \int_{r=0}^{\infty}\tilde{\lambda}_{k}\left(r\right)\cdot\mathrm{e}^{-\sum_{l=1}^{K}\lambda_{ol}\pi\left(P_{ol}\mathbb{E}\left[\Psi_{l}\right]B_{ol}\right)^{\frac{2}{\varepsilon_{l}}}r^{\frac{2}{\varepsilon_{l}}}}dr,\label{eq:pmfServingBS}\\
 &  & f_{T,\tilde{R}_{T,1}}\left(k,r\right)=\nonumber \\
 &  & \tilde{\lambda}_{k}\left(r\right)\cdot\mathrm{e}^{-\sum_{l=1}^{K}\lambda_{ol}\pi\left(P_{ol}\mathbb{E}\left[\Psi_{l}\right]B_{ol}\right)^{\frac{2}{\varepsilon_{l}}}r^{\frac{2}{\varepsilon_{l}}}},\label{eq:pdfServeBSMARP}
\end{eqnarray}
 for $k=1,\cdots,\ K$, \textup{where }$\tilde{\lambda}_{k}\left(r\right)$and
$\tilde{R}_{T,1}$ are from Corollary \ref{thm:stEquivalenceMARP}.
When $\left\{ \varepsilon_{l}\right\} _{l=1}^{K}=\varepsilon$, 
\begin{eqnarray}
\mathbb{P}\left(\left\{ T=k\right\} \right) & = & \frac{\lambda_{ok}\left(P_{ok}\mathbb{E}\left[\Psi_{k}\right]B_{ok}\right)^{\frac{2}{\varepsilon_{k}}}}{\sum_{l=1}^{K}\lambda_{ol}\left(P_{ol}\mathbb{E}\left[\Psi_{l}\right]B_{ol}\right)^{\frac{2}{\varepsilon_{l}}}}.\label{eq:pmfServingBSSplCase}
\end{eqnarray}
\end{lem}
\begin{proof}
Along the same lines as \cite[Lemmas 3 and 4]{Madhusudhanan2012},
the p.m.f. of the serving BS's tier-index and the c.c.d.f. of the
distance of the serving BS belonging to the $k^{\mathrm{th}}$ open-access
tier are derived below. 
\begin{eqnarray}
 &  & \mathbb{P}\left(\left\{ T=k\right\} \right)\overset{\left(a\right)}{=}\mathbb{P}\left(\bigcap_{l=1,\ l\ne k}^{K}\left\{ \tilde{R}_{l1}>\tilde{R}_{k1}\right\} \right)\nonumber \\
 &  & =\mathbb{E}_{\tilde{R}_{k1}}\left[\prod_{l=1,\ l\ne k}^{K}\mathbb{P}\left(\left.\left\{ \tilde{R}_{l1}>\tilde{R}_{k1}\right\} \right|\tilde{R}_{k1}\right)\right]\nonumber \\
 &  & \overset{\left(b\right)}{=}\mathbb{E}_{\tilde{R}_{k1}}\left[\mathrm{e}^{-\sum_{l=1,\ l\ne k}^{K}\int_{s=0}^{t}\tilde{\lambda}_{l}\left(s\right)ds}\right],\label{eq:proofPmfTierIndex}\\
 &  & \mathbb{P}\left(\left\{ T=k\right\} \bigcap\left\{ \tilde{R}_{k1}>r\right\} \right)\nonumber \\
 &  & =\mathbb{P}\left(\bigcap_{l=1,\ l\ne k}^{K}\left\{ \tilde{R}_{l1}>\tilde{R}_{k1}\right\} \bigcap\left\{ \tilde{R}_{k1}>r\right\} \right)\nonumber \\
 &  & =\mathbb{E}_{\tilde{R}_{k1}}\left[\mathcal{I}\left(\tilde{R}_{k1}>r\right)\mathbb{P}\left(\left.\bigcap_{l=1,\ l\ne k}^{K}\left\{ \tilde{R}_{l1}>\tilde{R}_{k1}\right\} \right|\tilde{R}_{k1}\right)\right]\nonumber \\
 &  & =\int_{t=r}^{\infty}\tilde{\lambda}_{k}\left(t\right)\mathrm{e}^{-\sum_{l=1}^{K}\int_{s=0}^{t}\tilde{\lambda}_{l}\left(s\right)ds}dt,\label{eq:ccdfServeBSMARP}
\end{eqnarray}
where (\ref{eq:proofPmfTierIndex}) - (a) is obtained since serving
BS belongs to the $k^{\mathrm{th}}$ tier if nearest BS among all
the open-access tiers in the equivalent hetnet of Corollary \ref{thm:stEquivalenceMARP}
belongs to the $k^{\mathrm{th}}$ tier, (\ref{eq:proofPmfTierIndex})
- (b) is obtained by noting that $\left\{ \tilde{R}_{l1}\right\} _{l=1}^{K}$
is a set of independent random variables with the p.d.f. of $\tilde{R}_{l1}$
as $f_{\tilde{R}_{l1}}\left(r\right)=\tilde{\lambda}_{l}\left(r\right)\cdot\mathrm{e}^{-\int_{s=0}^{r}\tilde{\lambda}_{l}\left(s\right)ds}$,
$r\ge0$ from the properties of Poisson process, and finally (\ref{eq:ccdfServeBSMARP})
is obtained. Using the steps in (\ref{eq:ccdfServeBSMARP}), (\ref{eq:pdfServeBSMARP})
is derived.
\end{proof}

When $\left\{ \mathbb{E}\left[\Psi_{l}\right]\right\} _{l=1}^{K}=1$,
(\ref{eq:pmfServingBS}) and (\ref{eq:pmfServingBSSplCase}) reduce
to \cite[(3) and (4)]{Jo2011}, respectively. Deriving the coverage
probability expressions for the arbitrary fading distribution case
under MBRP suffers from similar analytical intractabilities as the
nearest-BS case studied in Section \ref{sec:hetnetCoverage}. Hence,
we consider the special case where fading coefficients are i.i.d.
unit mean exponential random variables. In \cite{Jo2011}, Jo et.
al. have demonstrated that simple expressions for the hetnet coverage
probability under MARP can be computed when the fading coefficients
are i.i.d. exponential random variables. These results were restricted
to the open-access case, and are extended for a general hetnet below. 
\begin{thm}
\label{thm:covprobMARPExpFading}The hetnet coverage probability under
MBRP connectivity with i.i.d. exponential fading distribution at all
BSs is 
\begin{eqnarray}
 &  & \mathbb{P}_{\mathrm{coverage}}^{\mathrm{MBRP}}=\sum_{k=1}^{K}\lambda_{k}P_{k}^{\frac{2}{\varepsilon_{k}}}\beta_{k}^{-\frac{2}{\varepsilon_{k}}}\int_{r=0}^{\infty}2\pi r\mathrm{e}^{-\eta r^{\varepsilon_{k}}}\times\nonumber \\
 &  & \mathrm{e}^{-\frac{\lambda_{c}\pi P_{c}^{\frac{2}{\varepsilon_{c}}}r^{\frac{2\varepsilon_{k}}{\varepsilon_{c}}}}{\mathrm{sinc}\left(\frac{2\pi}{\varepsilon_{c}}\right)}-\sum_{l=1}^{K}\lambda_{l}\pi P_{l}^{\frac{2}{\varepsilon_{l}}}F\left(\beta_{k},\varepsilon_{l}\right)r^{\frac{2\varepsilon_{k}}{\varepsilon_{l}}}}dr\label{eq:covprobMARPexp}
\end{eqnarray}
where $F\left(\beta_{k},\varepsilon_{l}\right)=\frac{1}{\mathrm{sinc}\left(\frac{2\pi}{\varepsilon_{l}}\right)}+\beta_{k}^{-\frac{2}{\varepsilon_{l}}}\left[1-\ _{2}F_{1}\left(1,\frac{2}{\varepsilon_{l}};1+\frac{2}{\varepsilon_{l}};-\beta_{k}^{-1}\right)\right]$.
When $\left\{ \varepsilon_{k}\right\} _{k=1}^{K}=\varepsilon$ and
$\eta=0$, 
\begin{eqnarray}
 &  & \mathbb{P}_{\mathrm{coverage}}^{\mathrm{nearest}}=\mathbb{P}_{\mathrm{coverage}}^{\mathrm{MARP}}=\nonumber \\
 &  & \sum_{k=1}^{K}\frac{\lambda_{k}P_{k}^{\frac{2}{\varepsilon}}\beta_{k}^{-\frac{2}{\varepsilon}}\mathrm{sinc}\left(\frac{2\pi}{\varepsilon}\right)}{\lambda_{c}P_{c}^{\frac{2}{\varepsilon}}+\sum_{l=1}^{K}\lambda_{l}P_{l}^{\frac{2}{\varepsilon}}F\left(\beta_{k},\varepsilon\right)\mathrm{sinc}\left(\frac{2\pi}{\varepsilon}\right)}.\label{eq:covprobMARPSpCase1}
\end{eqnarray}
\end{thm}
\begin{proof}
See Appendix \ref{sub:proofCovprobMARPExpFading}.
\end{proof}
The above is a generalization of \cite[Theorem 1]{Jo2011} to closed-access
case. Further, comparing (\ref{eq:covprobExpBetaGt1NoNoise}) and
(\ref{eq:covprobMARPSpCase1}), clearly, $\mathbb{P}_{\mathrm{coverage}}^{\mathrm{MIRP}}\ge\mathbb{P}_{\mathrm{coverage}}^{\mathrm{MARP}}$,
when $\left\{ \beta_{k}\right\} _{k=1}^{\infty}\ge1$ since $F\left(\beta_{k},\varepsilon_{l}\right)\mathrm{sinc}\left(\frac{2\pi}{\varepsilon_{l}}\right)\ge1$,
$\forall\ \beta_{k}\ge0$, $\varepsilon_{l}>2$.

\section{Numerical Examples and Discussion}

In this section, we provide some numerical examples that complement
the theoretical results presented until now. We restrict ourselves
to the study of a two tier hetnet consisting of the macrocell and
the femtocell networks, respectively, under the max-SINR connectivity
model while reminding the reader that the theory presented in this
paper allows a similar analysis for arbitrary number of tiers and
also carries over to the nearest-BS connectivity model. Also, please
refer to Appendix \ref{sub:Simulation-Method} for the algorithm to
perform the Monte-Carlo simulations. For all the studies in this paper,
$\lambda_{2}=5\lambda_{1},$ $P_{1}=25P_{2},$ $\varepsilon=3,$ and
$\beta_{2}=1\ \mathrm{dB}$, where the subscripts `1' and `2' correspond
to macrocell and femtocell networks, respectively. Further, under
the closed-access BS association scheme, the MS has access to the
macrocell network only.
\begin{figure}
\begin{centering}
\includegraphics[scale=0.7]{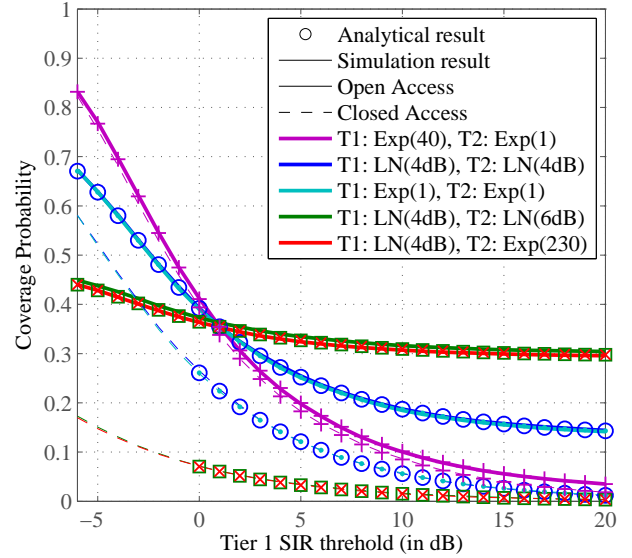}
\par\end{centering}

\caption{\textcolor{black}{\label{fig:covProbFigure}Two-tier hetnet: Comparing
coverage probabilities for various shadow fading distributions}}
\end{figure}
\begin{figure}
\begin{centering}
\includegraphics[clip,scale=0.7]{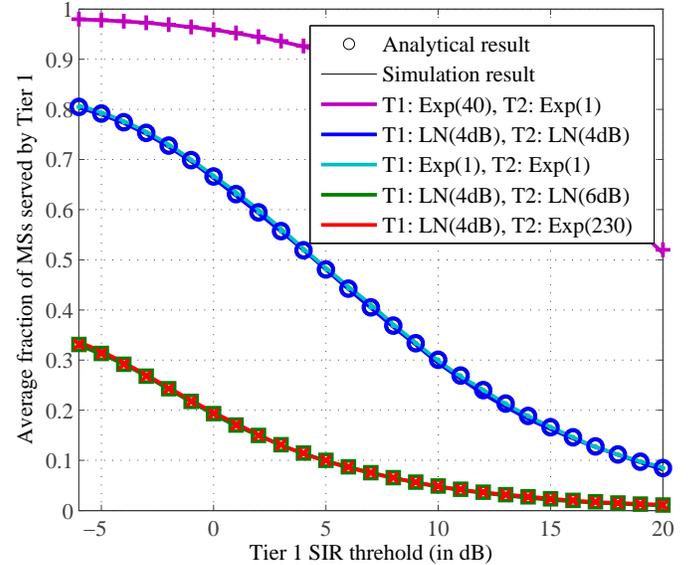}
\par\end{centering}

\caption{\label{fig:avgFracServedByT1}Two-tier hetnet: Average fraction of
MSs served by macrocell BSs vs macrocell SIR threshold }
\end{figure}
\begin{figure}
\begin{centering}
\includegraphics[scale=0.7]{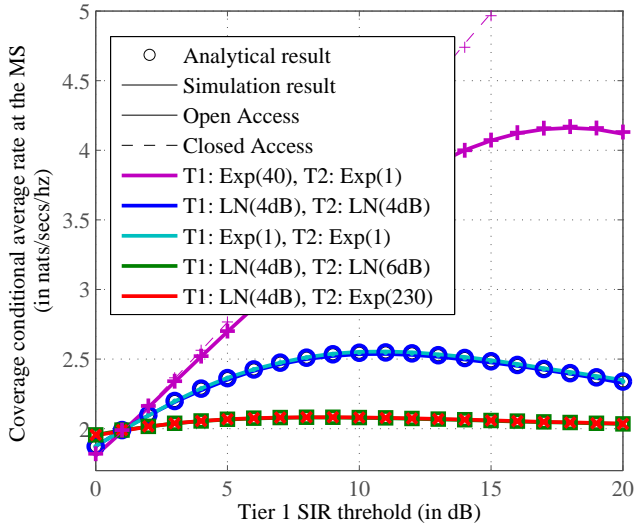}
\par\end{centering}

\caption{\label{fig:covCondAvgRate}Two-tier hetnet: Variation of coverage
conditional average rate with Tier 1 SIR threshold and different shadow
fading distributions}
\end{figure}

In Figures \ref{fig:covProbFigure}, \ref{fig:avgFracServedByT1}
and \ref{fig:covCondAvgRate}, we study the coverage probability,
coverage conditional average rate and the average fraction of users
served by the macrocell network, respectively, for various configurations
of shadow fading distributions at the macrocell and the femtocell
BSs. Note that the expressions for the coverage conditional average
rate and the average fraction of users served by the macrocell network
can be found in \cite[Theorems 2, 3 and 4]{Madhusudhanan2012c}. In
all the figures, T1 (T2) stand for tier 1, i.e. the macrocell network
(tier 2, i.e. the femtocell network). Further, Exp($\cdot$) and LN($\cdot$)
are abbreviations for exponential distribution with a given mean and
log-normal distribution with a zero mean and standard deviation (when
the random variable is expressed in dB), respectively, and they represent
distribution of the shadow fading factors of the corresponding tiers.

While the expressions in Theorem \ref{thm:covprobMIRPBetaGt1} clearly
show that a MS has a better coverage probability under open-access
than closed-access, the plots in Figure \ref{fig:covProbFigure} provides
a quantitative justification for the same. The coverage probability
curve corresponding to the exponential fading distribution at both
the tiers 1 and 2 with means 40 and 1, respectively, also corresponds
to the case where $P_{1}=1000P_{2},$ with the shadow fading factors
at both the tiers being unit mean exponential distributions. The open
and closed access have approximately the same coverage probabilities
because the MS is almost always served by a macrocell BS, as can been
seen in the corresponding curve in Figure \ref{fig:avgFracServedByT1}.
As a result, blocking access to the femtocell BSs altogether, has
only a marginal influence on the coverage probability at the MS. 

The two curves following the aforementioned curve in Figures \ref{fig:covProbFigure}-\ref{fig:covCondAvgRate}
complement the fact that all the three performance metrics are identical
irrespective of the distribution of the shadow fading factors, when
the shadow fading factors have the same distribution across all the
tiers. The last two curves in Figures \ref{fig:covProbFigure}-\ref{fig:covCondAvgRate}
show that all the performance metrics are identical as long as the
shadow fading coefficients of the corresponding tiers have the same
$\left(2/\varepsilon\right)^{\mathrm{th}}$ moments. Note that $\mathbb{E}\left[\Psi^{\frac{2}{\varepsilon}}\right]$
is the same when $\Psi$ has a log-normal distribution with zero mean
and 6 dB standard deviation or when $\Psi$ is an exponential random
variable with mean 230.

A log-normal random variable with zero mean and a given standard deviation
is a good model for shadow fading factors. Note that the femtocell
network is introduced to improve the indoor performance. The shadow
fading factors in the indoor environments are known to have a comparable
or greater standard deviation than otherwise. Such a situation is
represented by last four curves in Figures \ref{fig:covProbFigure}-\ref{fig:covCondAvgRate}.
The gap between the open and closed access coverage probability curves
indicate the contribution of the femtocell network in providing coverage
to the MS. It is immediately clear that the dense low-power femtocell
network has a more critical role in providing coverage in realistic
indoor models, when we look at the last four curves in Figures \ref{fig:covProbFigure}
and \ref{fig:avgFracServedByT1}.

Under open-access, the coverage probability and the coverage conditional
average rate (see Figures \ref{fig:covProbFigure} and \ref{fig:avgFracServedByT1})
for all the 5 curves mentioned above intersect when the SIR threshold
for the macrocell network is equal to 1 dB. This brings us to an important
point that when the SIR threshold is the same for all the tiers, these
metrics become independent of the transmission power and shadow fading
factors of the different tiers, and collapses to the corresponding
metrics in a single-tier network with the same path-loss exponent
and SIR threshold. Along the same lines, the coverage conditional
average rate for a two-tier hetnet under closed-access also collapses
to that of a single-tier network, and is independent of the transmission
power and shadow fading factors of the different tiers.

\section{Conclusions}

In this paper, for the most general model of the hetnets, the downlink
coverage probability and other related performance metrics such as
the average downlink rate and average fraction of users served by
each tier of the hetnet are characterized. Two important BS connectivity
models are studied, namely, the max-SINR and the nearest-BS connectivity,
respectively. Semi-analytical expressions for the hetnet coverage
probability is obtained for both the cases. Further, several properties
pertaining to the hetnet downlink performance are analyzed, which
provide great insights about these complex networks. As an example,
we identify the MIRP and MBRP connectivity models to be equivalent
to the former models under certain special conditions. These models
are much simpler to analyze and the results for these models expose
interesting properties of the hetnet. The results in this paper greatly
generalize the existing hetnet performance characterization results
and are essential for better understanding of the future developments
in wireless communications that are heavily based on hetnets.

\appendix

\subsection{\label{sub:proofLTmaxSINRnearestBS}Proof for Lemma \ref{lem:LTmaxSINRnearestBS}}

The proof for (\ref{eq:maxSINRLTexp}) is shown in (\ref{eq:maxSINRLTexpSteps})
\begin{figure*}
\begin{eqnarray}
 &  & \mathcal{L}_{I_{o}+I_{c}+\eta,\ \underset{k=1,\cdots,K}{\max}\gamma_{k}M_{k}\le u}\left(s\right)\nonumber \\
 &  & =\mathbb{E}\left[\exp\left(-s\left(I_{o}+I_{c}+\eta\right)\right)\times\mathcal{I}\left(\underset{k=1,\cdots,K}{\max}\gamma_{k}M_{k}\le u\right)\right]\nonumber \\
 &  & \overset{\left(a\right)}{=}\mathcal{L}_{I_{c}}\left(s\right)\mathrm{e}^{-s\eta}\mathbb{E}\left[\prod_{k=1}^{K}\prod_{l=1}^{\infty}\mathrm{e}^{-sP_{k}\Psi_{kl}R_{kl}^{-\varepsilon_{k}}}\mathcal{I}\left(\gamma_{k}P_{k}\Psi_{kl}R_{kl}^{-\varepsilon_{k}}\le u\right)\right]\nonumber \\
 &  & \overset{\left(b\right)}{=}\mathcal{L}_{I_{c}}\left(s\right)\mathrm{e}^{-s\eta}\prod_{k=1}^{K}\mathbb{E}\left[\prod_{l=1}^{\infty}\mathrm{e}^{-sP_{k}\Psi_{kl}R_{kl}^{-\varepsilon_{k}}}\mathcal{I}\left(P_{k}\Psi_{kl}R_{kl}^{-\varepsilon_{k}}\le\frac{u}{\gamma_{k}}\right)\right]\nonumber \\
 &  & \overset{\left(c\right)}{=}\mathcal{L}_{I_{c}}\left(s\right)\mathrm{e}^{-s\eta}\prod_{k=1}^{K}\exp\left(-\lambda_{k}\int_{r=0}^{\infty}\left(1-\mathbb{E}\left[\mathrm{e}^{-sP_{k}\Psi_{k}r^{-\varepsilon_{k}}}\mathcal{I}\left(P_{k}\Psi_{k}r^{-\varepsilon_{k}}\le\frac{u}{\gamma_{k}}\right)\right]\right)2\pi rdr\right)\nonumber \\
 &  & \overset{\left(d\right)}{=}\mathcal{L}_{I_{c}}\left(s\right)\mathrm{e}^{-s\eta}\prod_{k=1}^{K}\exp\left(-\lambda_{k}\mathbb{E}_{\Psi_{k}}\left[\int_{t=0}^{\infty}\left(1-\mathrm{e}^{-t}\mathcal{I}\left(t\le\frac{su}{\gamma_{k}}\right)\right)\frac{2\pi}{\varepsilon_{k}}t^{-\frac{2}{\varepsilon}-1}\left(sP_{k}\Psi_{k}\right)^{\frac{2}{\varepsilon_{k}}}dt\right]\right)\nonumber \\
 &  & \overset{\left(e\right)}{=}\mathcal{L}_{I_{c}}\left(s\right)\mathrm{e}^{-s\eta}\prod_{k=1}^{K}\exp\left(-\lambda_{k}\pi\left(sP_{k}\right)^{\frac{2}{\varepsilon_{k}}}\mathbb{E}\left[\Psi_{k}^{\frac{2}{\varepsilon_{k}}}\right]\left[\Gamma\left(1-\frac{2}{\varepsilon_{k}}\right)+\frac{2}{\varepsilon_{k}}\int_{t=0}^{\infty}\mathrm{e}^{-t}t^{-\frac{2}{\varepsilon}-1}\mathcal{I}\left(t>\frac{su}{\gamma_{k}}\right)dt\right]\right).\label{eq:maxSINRLTexpSteps}\\
\nonumber \\
 &  & \mathcal{L}_{I_{o}+I_{c}+\eta,\ \underset{k=1,\cdots,K}{\max}\gamma_{k}N_{k}\le u}\left(s\right)\nonumber \\
 &  & \overset{\left(a\right)}{=}\mathcal{L}_{I_{c}}\left(s\right)\mathrm{e}^{-s\eta}\mathbb{E}\left[\prod_{k=1}^{K}\prod_{l=1}^{\infty}\mathrm{e}^{-sP_{k}\Psi_{kl}R_{kl}^{-\varepsilon_{k}}}\times\mathcal{I}\left(\underset{k=1,\cdots,K}{\max}\gamma_{k}P_{k}\Psi_{k1}R_{k1}^{-\varepsilon_{k}}\le u\right)\right]\nonumber \\
 &  & =\mathcal{L}_{I_{c}}\left(s\right)\mathrm{e}^{-s\eta}\mathbb{E}\left[\prod_{k=1}^{K}\prod_{l=1}^{\infty}\mathrm{e}^{-sP_{k}\Psi_{kl}R_{kl}^{-\varepsilon_{k}}+\mathrm{ln}\left(\mathcal{I}\left(\gamma_{k}P_{k}\Psi_{k1}R_{k1}^{-\varepsilon_{k}}\le u\right)\right)}\right]\nonumber \\
 &  & \overset{\left(b\right)}{=}\mathcal{L}_{I_{c}}\left(s\right)\mathrm{e}^{-s\eta}\prod_{k=1}^{K}\mathbb{E}_{\Psi_{k1},R_{k1}}\left[\mathrm{e}^{-sP_{k}\Psi_{k1}R_{k1}^{-\varepsilon_{k}}}\mathcal{I}\left(\gamma_{k}P_{k}\Psi_{k1}R_{k1}^{-\varepsilon_{k}}\le u\right)\mathbb{E}\left[\left.\prod_{l=2}^{\infty}\mathrm{e}^{-sP_{k}\Psi_{kl}R_{kl}^{-\varepsilon_{k}}\mathcal{I}\left(R_{kl}>R_{k1}\right)}\right|R_{k1}\right]\right]\nonumber \\
 &  & \overset{\left(c\right)}{=}\mathcal{L}_{I_{c}}\left(s\right)\mathrm{e}^{-s\eta}\prod_{k=1}^{K}\mathbb{E}_{\Psi_{k1},R_{k1}}\left[\mathrm{e}^{-sP_{k}\Psi_{k1}R_{k1}^{-\varepsilon_{k}}}\mathcal{I}\left(\gamma_{k}P_{k}\Psi_{k1}R_{k1}^{-\varepsilon_{k}}\le u\right)\mathrm{e}^{-\lambda_{k}\int_{r=R_{k1}}^{\infty}\left(1-\mathbb{E}\left[\mathrm{e}^{-sP_{k}\Psi_{k}r^{-\varepsilon_{k}}}\right]\right)2\pi rdr}\right]\nonumber \\
 &  & \overset{\left(d\right)}{=}\mathcal{L}_{I_{c}}\left(s\right)\mathrm{e}^{-s\eta}\prod_{k=1}^{K}\mathbb{E}_{\Psi_{k1},R_{k1}}\left[\mathrm{e}^{-sP_{k}\Psi_{k1}R_{k1}^{-\varepsilon_{k}}}\mathcal{I}\left(\gamma_{k}P_{k}\Psi_{k1}R_{k1}^{-\varepsilon_{k}}\le u\right)\times\right.\nonumber \\
 &  & \left.\mathrm{e}^{-\lambda_{k}\pi\left(sP_{k}\right)^{\frac{2}{\varepsilon_{k}}}\mathbb{E}_{\Psi_{k}}\left[\Psi_{k}^{\frac{2}{\varepsilon_{k}}}\int_{t=0}^{sP_{k}\Psi_{k}R_{k1}^{-\varepsilon_{k}}}\left(1-\mathrm{e}^{-t}\right)\frac{2}{\varepsilon_{k}}t^{-\frac{2}{\varepsilon_{k}}-1}dt\right]}\right],\label{eq:proofOutlineNearestBSLTExp}
\end{eqnarray}

\rule[0.5ex]{2\columnwidth}{1pt}
\end{figure*}
 where $(a)$ is obtained by noting that $I_{c}$ is independent of
the random variables $I_{o}$ and $\underset{k=1,\cdots,K}{\max}\gamma_{k}M_{k}\le u$,
$\mathcal{L}_{I_{c}}\left(s\right)$ is a direct consequence of the
Campbell's theorem \cite{Kingman1993}, $\mathrm{e}^{-s\eta}$ is
a constant and $\left\{ \underset{k=1,\cdots,K}{\max}\gamma_{k}M_{k}\le u\right\} \iff\left\{ \gamma_{k}P_{k}\Psi_{kl}R_{kl}^{-\varepsilon_{k}}\le u\right\} $,
$\forall\ k=1,\cdots,\ K$ and $l=1,\ 2,\cdots$; $(b)$ is obtained
since the random variables corresponding to a given tier are independent
of the other tiers; $(c)$ is obtained by applying the Campbell's
theorem \cite{Kingman1993} to each tier of the hetnet; $(d)$ is
obtained by changing the variable of integration from $r$ to $t=sP_{k}\Psi_{k}r^{-\varepsilon_{k}}$;
$(e)$ is obtained by rewriting the integral in $(d)$ using special
functions; and finally (\ref{eq:maxSINRLTexp}) is obtained by rewriting
the integral in $(e)$ in terms of the incomplete Gamma function.

The proof for (\ref{eq:nearestBSLTexp}) follows along the same lines
as above and we provide only a brief outline in (\ref{eq:proofOutlineNearestBSLTExp})
where the maximization in $(a)$ is only among the nearest BSs of
the K tiers of the hetnet, $\mathcal{L}_{I_{c}}\left(s\right)$ is
the same as in (\ref{eq:maxSINRLTexpSteps}); $\left(b\right)$ is
obtained by exchanging the order of expectation and product since
the K tiers of the hetnet are independent of each other, and further
conditioning w.r.t. the fading coefficient and the distance of the
nearest BS of each tier; $\left(c\right)$ is obtained by applying
the Campbell's theorem to the set of $k^{\mathrm{th}}$ tier BSs beyond
$R_{k1}$, conditioned on $R_{k1}$; $\left(d\right)$ is obtained
by further simplifying $\left(c\right)$; and finally (\ref{eq:nearestBSLTexp})
is obtained by evaluating the expectation w.r.t. $R_{k1}$ in $\left(d\right)$
where the p.d.f. of $R_{k1}$ is $f_{R_{k1}}\left(r\right)=\lambda_{k}2\pi r\mathrm{e}^{-\lambda_{k}\pi r^{2}},\ r\ge0$,
and further simplifying.

\subsection{\label{sub:proofStEquivalenceLemma1}Proof for Lemma \ref{lem:StochasticEqLemma1}}

Given a BS belonging to the $k^{\mathrm{th}}$ open-access tier is
at a distance $R_{k}$ from the origin, then, due to \cite[Theorem 2]{Madhusudhanan2012},
$\left.\tilde{R}\right|k=\left(P_{k}\Psi_{k}\right)^{-1}R_{k}^{\varepsilon_{k}}$
represents the distance of the BS from the origin where the BS arrangement
is according to a non-homogeneous 1-D Poisson point process with BS
density function $\lambda^{\left(k\right)}\left(r\right)$, as long
as $\mathbb{E}\left[\Psi_{k}^{\frac{2}{\varepsilon_{k}}}\right]<\infty$,
for each $k=1,\ 2,\cdots,\ K$. Similarly, for the closed-access tier,
$\hat{R}=\left(P_{c}\Psi_{c}\right)^{-1}R_{c}^{\varepsilon_{c}}$
the distance where the BS arrangement is according to a non-homogeneous
1-D Poisson point process with BS density function $\hat{\lambda}\left(r\right)$,
as long as $\mathbb{E}\left[\Psi_{c}^{\frac{2}{\varepsilon_{c}}}\right]<\infty$.
This is a consequence of the Mapping theorem \cite[Page 18]{Kingman1993}
and the Marking Theorem \cite[Page 55]{Kingman1993} of the Poisson
processes. Further, since the BS arrangements in the different tiers
were originally independent of each other, the set of all the BSs
in the equivalent 1-D non-homogeneous Poisson process is merely the
union of all $\left.\tilde{R}'s\right|k,\ \forall\ k=1,\ 2,\cdots,\ K.$
By the Superposition Theorem \cite[Page 16]{Kingman1993} of Poisson
process, $\tilde{R}$ (notice that it is not conditioned on $k$)
corresponds to the distance from origin of BS arrangement according
to non-homogeneous Poisson point process with density function $\tilde{\lambda}\left(r\right)=\sum_{k=1}^{K}\lambda^{\left(k\right)}\left(r\right),\ r\ge0.$ 

In summary, we have converted the BS arrangement on a 2-D plane of
hetnet to a BS arrangement of the equivalent 2-tier network along
1-D (positive x-axis), and hence, the SINR distributions of both these
networks are also equivalent. Further, by our construction, the MIRP
BS in the hetnet corresponds to the BS that is nearest to the origin
(MS) in the equivalent 2-tier network. As a result, SINR may be written
in terms of the $\tilde{R}$'s and $\hat{R}$'s indexed in the ascending
order, and we get $\left(\ref{eq:SINRDistributionEq1}\right).$

\subsection{\label{sub:proofStEquivalenceLemma2}Proof for Lemma \ref{lem:StochasticEqLemma2}}

The hetnet SINR under MIRP can be computed as follows. For each tier
$m=1,\cdots,\ K,\ c$ (c refers to the closed-access tier), form the
set $\left\{ \left(P_{m}\Psi_{m,l}\right)^{-\frac{1}{\varepsilon}}R_{m,l}\right\} _{l=1}^{\infty}$
and represent as $\left\{ \bar{R}_{m,k}\right\} _{k=1}^{\infty}$
where $\tilde{R}$'s are ascendingly ordered. Now, $\left\{ \bar{R}_{m,k}^{-\varepsilon}\right\} _{k=1}^{\infty}$
represents the received powers of all the $m^{\mathrm{th}}$ tier
BSs in the descending order. Finally, the desired BS's power and tier
index $\left(T\right)$ can be easily found by identifying the maximum
in the set $\left\{ \bar{R}_{m,1}^{-\varepsilon}\right\} _{m=1}^{K}$
and the SINR can be computed. Using \cite[Corollary 3]{Madhusudhanan2010a}
which is an application of the Marking theorem \cite[Page 55]{Kingman1993},
it can be seen that $\left\{ \bar{R}_{m,k}\right\} _{k=1}^{\infty}$
represents the distances from origin of BSs arranged according to
homogeneous Poisson point process with BS density $\lambda_{m}P_{m}\mathbb{E}\left[\Psi_{m}^{\frac{2}{\varepsilon}}\right],$
where $\Psi_{m}$ has the same distribution as the $m^{\mathrm{th}}$
tier shadow fading factors. As a result, the set $\left\{ \bar{R}_{m,l}^{-\varepsilon}\right\} _{m=1,\ l=1}^{m=K,\ l=\infty}$
represents the set of received powers at the origin of the hetnet
composed of $K$ open-access tiers and a closed-acess tier with BS
densities $\left\{ \lambda_{k}P_{k}\mathbb{E}\left[\Psi_{k}^{\frac{2}{\varepsilon}}\right]\right\} _{k=1}^{K},\ \lambda_{c}P_{c}\mathbb{E}\left[\Psi_{c}^{\frac{2}{\varepsilon}}\right]$,
respectively, with unity transmit powers and shadow fading factors
at each BS. This is equivalent to the original heterogeneous network
and has the same SINR distribution, hence proving $\left(\ref{eq:SINRDistributionEq2}\right).$ 

Further, using the Superposition theorem \cite[Page 16]{Kingman1993},
the $K$ open-access tiers of the equivalent hetnet can be combined
to form a single tier network with a BS density equal to $\sum_{l=1}^{K}\lambda_{l}P_{l}^{\frac{2}{\varepsilon}}\mathbb{E}\Psi_{l}^{\frac{2}{\varepsilon}},$
thus proving the SINR equivalence in $\left(\ref{eq:SINRDistEqDensity1System2}\right).$
The distribution of SINR of this two-tier network is the same as that
of an MS in another two-tier network where the open-access tier has
unity BS density, the closed-access tier has a BS density $\frac{\lambda_{c}P_{c}\mathbb{E}\left[\Psi_{c}^{\frac{2}{\varepsilon}}\right]}{\sum_{l=1}^{K}\lambda_{l}P_{l}^{\frac{2}{\varepsilon}}\mathbb{E}\Psi_{l}^{\frac{2}{\varepsilon}}}$,
unity transmit power and shadow fading factors at all BSs and a background
noise $\frac{\eta}{\left(\sum_{l=1}^{K}\lambda_{l}P_{l}^{\frac{2}{\varepsilon}}\mathbb{E}\Psi_{l}^{\frac{2}{\varepsilon}}\right)^{-\frac{\varepsilon}{2}}}$,
due to \cite[Lemma 3]{Madhusudhanan2010a} and hence we get the relation
$\left(\ref{eq:SINRDistEqDensity1System}\right).$

\subsection{\label{sub:proofCovprobMIRPBetaGt1}Proof for Theorem \ref{thm:covprobMIRPBetaGt1}}

From Corollary \ref{cor:arbitraryFadingCorollary} and Lemma \ref{lem:StochasticEqLemma1},
we get the following stochastic equivalence: 
\begin{eqnarray*}
 &  & \mathrm{SINR}_{T,I}=_{\mathrm{st}}\\
 &  & \left.\frac{h_{T,I}\tilde{R}_{T,I}^{-1}}{\underset{\left(k,l\right)\ne\left(T,I\right)}{\sum_{k=1}^{K}\sum_{l=1}^{\infty}}h_{kl}\tilde{R}_{kl}^{-1}+\sum_{l=1}^{\infty}g_{l}\hat{R}_{l}^{-1}+\eta}\right|_{\left(\left\{ \tilde{\lambda}_{k}\left(r\right)\right\} _{k=1}^{\infty},\hat{\lambda}\left(r\right)\right)},
\end{eqnarray*}
where $h_{kl}$'s and $g_{l}$'s are i.i.d. unit mean exponential
random variables, $J=\underset{k=1,2,\cdots}{\mathrm{argmax}}\ h_{T,k}\tilde{R}_{T,k}^{-1}$,
$\left\{ \tilde{R}_{kl}\right\} _{l=1}^{\infty}$ and $\left\{ \hat{R}_{l}\right\} _{l=1}^{\infty}$
are from non-homogeneous 1-D Poisson processes with density functions
$\tilde{\lambda}_{k}\left(r\right)=\lambda_{k}\frac{2\pi}{\varepsilon_{k}}P_{k}^{\frac{2}{\varepsilon_{k}}}r^{\frac{2}{\varepsilon_{k}}-1}$,
$k=1,\cdots,\ K$ and $\hat{\lambda}\left(r\right)=\lambda_{c}\frac{2\pi}{\varepsilon_{c}}P_{c}^{\frac{2}{\varepsilon_{c}}}r^{\frac{2}{\varepsilon_{c}}-1}$,
respectively. The following steps derive the hetnet coverage probability
and closely follows the proof techniques for \cite[Theorem 4]{Madhusudhanan2012}
and \cite[Theorem 1]{Dhillon2012a} 
\begin{eqnarray*}
 &  & \mathbb{P}_{\mathrm{coverage}}^{\mathrm{max-SINR}}=\mathbb{P}_{\mathrm{coverage}}^{\mathrm{MIRP}}\\
 &  & =\sum_{i=1}^{K}\mathbb{P}\left(\left\{ \frac{h_{ij}\tilde{R}_{ij}^{-1}}{\underset{\left(k,l\right)\ne\left(i,j\right)}{\sum_{k=1}^{K}\sum_{l=1}^{\infty}}h_{kl}\tilde{R}_{kl}^{-1}+\sum_{l=1}^{\infty}g_{l}\hat{R}_{l}^{-1}+\eta}>\beta_{i}\right\} \right)\\
 &  & \overset{\left(a\right)}{=}\sum_{i=1}^{K}\mathbb{E}_{\tilde{R}_{ij}}\left[\mathrm{e}^{-\beta_{i}\tilde{R}_{ij}\eta}\mathbb{E}\left[\left.\mathrm{e}^{-\beta_{i}\tilde{R}_{ij}\underset{\left(k,l\right)\ne\left(i,j\right)}{\sum_{k=1}^{K}\sum_{l=1}^{\infty}}h_{kl}\tilde{R}_{kl}^{-1}}\right|\tilde{R}_{ij}\right]\times\right.\\
 &  & \left.\mathbb{E}\left[\left.\mathrm{e}^{-\beta_{i}\tilde{R}_{ij}\sum_{l=1}^{\infty}g_{l}\hat{R}_{l}^{-1}}\right|\tilde{R}_{ij}\right]\right]\\
 &  & \overset{\left(b\right)}{=}\sum_{i=1}^{K}\int_{r=0}^{\infty}\tilde{\lambda}_{i}\left(r\right)\mathrm{e}^{-\eta\beta_{i}r-\frac{\lambda_{c}\pi\left(P_{c}\beta_{i}r\right)^{\frac{2}{\varepsilon_{c}}}\mathbb{E}\left[\Psi_{c}^{\frac{2}{\varepsilon_{c}}}\right]}{\Gamma\left(1+\frac{2}{\varepsilon_{c}}\right)\mathrm{sinc}\left(\frac{2\pi}{\varepsilon_{c}}\right)}}\times\\
 &  & \mathrm{e}^{-\sum_{l=1}^{K}\frac{\lambda_{l}\pi\left(P_{l}\beta_{i}r\right)^{\frac{2}{\varepsilon_{l}}}\mathbb{E}\left[\Psi_{l}^{\frac{2}{\varepsilon_{l}}}\right]}{\Gamma\left(1+\frac{2}{\varepsilon_{l}}\right)\mathrm{sinc}\left(\frac{2\pi}{\varepsilon_{l}}\right)}}dr,
\end{eqnarray*}
where $\tilde{R}_{ij}$ is the distance from the origin of an arbitrary
point in the non-homogeneous Poisson process with density function
$\tilde{\lambda}_{i}\left(r\right)$, $\left(a\right)$ is obtained
by computing the probability of w.r.t. $h_{ij}$ conditioned on all
the other involved random variables and noting that the two Poisson
processes are independent of each other, $\left(b\right)$ is obtained
by evaluating the inner expectations by applying Campbell's theorem
\cite{Kingman1993} (same steps as in the proof of \cite[Theorem 4]{Madhusudhanan2012})
and expressing the expectation w.r.t. $\tilde{R}_{ij}$ by the integral
where $\tilde{\lambda}\left(r\right)dr$ is the probability that there
exists a point in the interval $\left(r,r+dr\right)$, and finally
(\ref{eq:covProbExpressionBetaGt1}) is obtained by simplifying the
integral in $\left(b\right)$.

\subsection{\label{sub:proofCovprobMARPExpFading}Proof for Theorem \ref{thm:covprobMARPExpFading}}

The the hetnet coverage probability is 
\begin{eqnarray*}
 &  & \mathbb{P}_{\mathrm{coverage}}^{\mathrm{MBRP}}\\
 &  & \overset{\left(a\right)}{=}\sum_{k=1}^{K}\mathbb{E}_{T,\tilde{R}_{T1}}\left[\mathbb{P}\left(\left\{ \tilde{P}_{k1}\Psi_{k1}\tilde{R}_{k1}^{-1}\left/\left(\sum_{n=1}^{\infty}\frac{\Psi_{cn}}{\hat{R}_{cn}}+\eta+\right.\right.\right.\right.\right.\\
 &  & \left.\left.\left.\left.\left.\underset{\left(m,l\right)\ne\left(k,1\right)}{\sum_{m=1}^{K}\sum_{l=1}^{\infty}}\frac{\tilde{P}_{ml}\Psi_{ml}\mathcal{I}\left(\tilde{R}_{ml}>\tilde{R}_{k1}\right)}{\tilde{R}_{ml}}\right)>\beta_{T}\right\} \right|k,\tilde{R}_{k1}\right)\right]\\
 &  & \overset{\left(b\right)}{=}\sum_{k=1}^{K}\mathbb{E}_{T,\tilde{R}_{T1}}\left[\mathrm{e}^{-\frac{\eta\beta_{k}\tilde{R}_{k1}}{\tilde{P}_{k1}}}\times\right.\\
 &  & \underset{=E_{1}}{\underbrace{\mathbb{E}\left[\left.\prod_{n=1}^{\infty}\mathrm{e}^{\frac{-\beta_{k}\tilde{R}_{k1}\Psi_{cn}}{\tilde{P}_{k1}\hat{R}_{cn}}}\right|T=k,\ \tilde{R}_{T1}=\tilde{R}_{k1}\right]}}\times\\
 &  & \left.\prod_{m=1}^{K}\underset{=E_{2}}{\underbrace{\mathbb{E}\left[\left.\prod_{l=1}^{\infty}\mathrm{e}^{\frac{-\beta_{k}\tilde{R}_{k1}\Psi_{ml}}{\tilde{P}_{k1}\tilde{R}_{ml}}\mathcal{I}\left(\tilde{R}_{ml}>\tilde{R}_{k1}\right)}\right|T=k,\ \tilde{R}_{T1}=\tilde{R}_{k1}\right]}}\right]
\end{eqnarray*}
where $\left(a\right)$ is from the stochastic equivalence in Corollary
\ref{thm:stEquivalenceMARP}, $\left(b\right)$ is obtained due to
the independence of each tier in the hetnet given $\left(T,R_{T1}\right)$.
Now, we derive expressions for $E_{1}$ and $E_{2}$ in (b). 
\begin{eqnarray*}
 &  & E_{1}=\exp\left(-\int_{r=0}^{\infty}\left(1-\mathbb{E}_{\Psi_{c}}\left[\mathrm{e}^{\frac{-\beta_{k}\tilde{R}_{k1}\Psi_{c}}{\tilde{P}_{k1}r}}\right]\right)\hat{\lambda}\left(r\right)dr\right)\\
 &  & =\mathrm{e}^{-\frac{\lambda_{c}\pi\left(P_{c}\beta_{k}\tilde{R}_{k1}\right)^{\frac{2}{\varepsilon_{c}}}}{\tilde{P}_{k1}^{\frac{2}{\varepsilon_{c}}}\mathrm{sinc}\left(\frac{2\pi}{\varepsilon_{c}}\right)}},\\
 &  & E_{2}=\exp\left(-\int_{r=\tilde{R}_{k1}}^{\infty}\left(1-\mathbb{E}_{\Psi_{m}}\left[\mathrm{e}^{\frac{-\beta_{k}\tilde{R}_{k1}\Psi_{m}}{\tilde{P}_{k1}r}}\right]\right)\tilde{\lambda}_{m}\left(r\right)dr\right)\\
 &  & =\mathrm{e}^{-\lambda_{om}\pi\frac{2\beta_{k}\left(P_{om}B_{om}\tilde{R}_{k1}\right)^{\frac{2}{\varepsilon_{m}}}}{\tilde{P}_{k1}\left(\varepsilon_{m}-2\right)}\times_{2}F_{1}\left(1,1-\frac{2}{\varepsilon_{m}};2-\frac{2}{\varepsilon_{m}};-\frac{\beta_{k}}{\tilde{P}_{k1}}\right)}
\end{eqnarray*}
Finally, (\ref{eq:covprobMARPexp}) is obtained by computing each
expectation in $\left(b\right)$ by applying the Campbell-Mecke theorem. 

For $\left\{ \varepsilon_{k}\right\} _{k=1}^{K}=\varepsilon$ and
$\eta=0$, the integral in (\ref{eq:covprobMARPexp}) simplifies to
(\ref{eq:covprobMARPSpCase1}).

\subsection{\label{sub:Simulation-Method}Simulation Method}

The $k^{\mathrm{th}}$ tier of the hetnet with $K$ tiers is identified
by the following set of system parameters: $\left(\lambda_{k},\ P_{k},\ \Psi_{k},\ \varepsilon_{k},\ \beta_{k}\right),$
where the symbols have all been defined in Section \ref{sec:System-Model},
and $k=1,\ 2,\cdots,\ K,$ where $K$ is the total number of tiers.
Now we illustrate the steps for simulating the hetnet in order to
obtain the SINR distribution and the coverage probability assuming
the MS is at the origin. The algorithm for the Monte-Carlo simulation
is as follows:

1) Generate $N_{k}$ random variables according to a uniform distribution
in the circular region of radius $R_{B}$ for the locations of all
the $k^{\mathrm{th}}$ tier BSs, where $N_{k}\sim\mathrm{Poisson}\left(\lambda_{k}\pi R_{B}^{2}\right).$

3) Compute the SINR at the desired BS according to Section \ref{sub:Cell-Association-policy}
and record the tier index $I$ of the desired BS.

Repeat the same procedure T (typically, > 50000) times. Finally, the
tail probability of SINR at $\eta$ is given by $\frac{\left\{ \mbox{\# of trials where SINR >}\ \eta\right\} }{T},$
and the coverage probability is given by $\sum_{k=1}^{K}\frac{\left\{ \mbox{\# of trials where}\ I=k\ \mathrm{and}\ SINR>\beta_{k}\right\} }{T}$.

\bibliographystyle{IEEEtran}
\bibliography{prasanna}

\end{document}